\pdfoutput=1
\documentclass{CSML}
\pdfoutput=1

\usepackage{lastpage}
\newcommand{\dOi}{14(2:8)2018}
\lmcsheading{}{1--\pageref{LastPage}}{}{}%
{Feb.~22, 2017}{May~09, 2018}{}

\usepackage{stackrel}
 \usepackage{wrapfig}
\usepackage{verbatim}
\usepackage{graphicx}
\usepackage{color}
\usepackage{xspace}
\usepackage{epsfig}
\usepackage{savesym}
\savesymbol{iint}
\savesymbol{iiint}
\usepackage{amsmath,amssymb}%,shuffle}%,amsfonts}
\usepackage{wasysym}
\usepackage{pstricks}
\usepackage{stmaryrd}
\usepackage{enumitem}

\usepackage{thm-restate}
\usepackage{pgf}

\usepackage{qtree}

\usepackage[colorlinks=true,citecolor=red!40!black,linkcolor=blue!50!black,
urlcolor=green!50!black,bookmarksopen=true,bookmarksopenlevel=2
]{hyperref}
\hypersetup{hidelinks}

\usepackage{todonotes}

\usepackage{tikz}
\usetikzlibrary{fit} % fitting shapes to coordinates
\usetikzlibrary{backgrounds} % drawing the background after the foreground
\usetikzlibrary{patterns,calc}
\usetikzlibrary{decorations.pathreplacing}
\usetikzlibrary{positioning}
\usetikzlibrary{decorations.pathmorphing}
\usetikzlibrary{decorations.markings}
\tikzstyle{background}=[rectangle,fill=gray!10, inner sep=0.1cm, rounded corners=0mm]
\usepgflibrary{shapes}
\usetikzlibrary{snakes,automata,arrows}
\usetikzlibrary{shadows}
\tikzstyle{nloc}=[draw, text badly centered, rectangle, rounded corners, minimum size=2em,inner sep=0.5em]
\tikzstyle{background}=[rectangle,fill=gray!10, inner sep=0.1cm, rounded corners=0mm]
\tikzstyle{loc}=[draw,rectangle,minimum size=1.4em,inner sep=0em]
\tikzset{
    gluon/.style={decorate,draw=black,
        decoration={coil,amplitude=1pt, segment length=5pt}} 
}

\tikzset{
    gluonew/.style={decorate,draw=black,
        decoration={coil,amplitude=1pt, segment length=2pt}} 
}

\tikzset{
    gluon1/.style={decorate,draw=black,
        decoration={coil,amplitude=3pt, segment length=3pt}} 
}

\theoremstyle{plain}

\newtheorem{claim}[thm]{Claim}
\newtheorem{definition}[thm]{Definition}
\newtheorem{lemma}[thm]{Lemma}
\newtheorem{theorem}[thm]{Theorem}
\newtheorem{corollary}[thm]{Corollary}

\newcommand{\AKMV}{\ensuremath{\mathcal{A}^{K,M}_{\textsf{valid}}}\xspace}

\newcommand{\AKMS}{\ensuremath{\mathcal{A}^{K,M}_{\Sys}}\xspace}
\newcommand{\TCW}{\ensuremath{\mathsf{TCW}}\xspace}
\newcommand{\TCWs}{\ensuremath{\mathsf{TCWs}}\xspace}
\newcommand{\STCW}{\ensuremath{\mathsf{STCW}}\xspace}
\newcommand{\STCWs}{\ensuremath{\mathsf{STCWs}}\xspace}
\newcommand{\KTCW}{\STCW^{K}\xspace}
\newcommand{\KMTCW}{\STCW^{K,M}\xspace}

\newcommand{\Sys}{\ensuremath{\mathcal{S}}\xspace}
\newcommand{\Lang}{\ensuremath{\mathcal{L}}\xspace}
\newcommand{\op}{\ensuremath{\mathsf{op}}\xspace}
\newcommand{\nop}{\ensuremath{\mathsf{nop}}\xspace}
\newcommand{\push}{\ensuremath{\downarrow}\xspace}
\newcommand{\pop}{\ensuremath{\uparrow}\xspace}

\newcommand{\EP}{\ensuremath{\mathsf{EP}}\xspace}

\newcommand{\atomicSTT}{\ensuremath{\mathsf{atomicSTT}}\xspace}
\newcommand{\STT}{\ensuremath{\mathsf{STT}}\xspace}
\newcommand{\STTs}{\ensuremath{\mathsf{STTs}}\xspace}
\newcommand{\kSTT}{$K$-\STT\xspace}
\newcommand{\kSTTs}{$K$-\STTs\xspace}

\newcommand{\stt}{\tau}
\newcommand{\add}[3]{\mathop{\mathsf{Add}_{#1,#2}^{#3}}}
\newcommand{\forget}[1]{\mathop{\mathsf{Forget}_#1}}

\newcommand{\stcw}{\ensuremath{\mathcal{V}}\xspace}
\newcommand{\Events}{V}

\newcommand{\labelroot}{\mathsf{labroot}}
\newcommand{\image}{\mathsf{Im}}

\newcommand{\rename}[2]{\mathop{\mathsf{Rename}_{#1,#2}}}
\newcommand{\sttunion}{\oplus}

\newcommand{\N}{\mathbb{N}}

\newcommand{\sem}[1]{\llbracket #1 \rrbracket}
\newcommand{\matchrel}{\curvearrowright}

\newcommand{\procrel}{\rightarrow}
\newcommand{\hole}{\dashrightarrow}

\newcommand{\width}{\ensuremath{\mathsf{width}}\xspace}

\newcommand{\cA}{{\mathcal{A}}} 
 
 \newcommand{\cL}{{\mathcal {L}}}
 
\newcommand{\cI}{{\mathcal{I}}} 
 
\newcommand{\cO}{{\mathcal {O}}}

\newcommand{\true}{\textsf{tt}} % à utiliser en mode mathématique
\newcommand{\false}{\textsf{ff}} % à utiliser en mode mathématique

\newcommand{\dom}{\ensuremath{\textsf{{\sffamily dom}}}}

\newcommand{\stc}{\stcw=(V,\procrel\cup{\hole},(\matchrel^I)_{I\in\mathcal{I}},\lambda)}

\newcommand{\MPDA}{\ensuremath{\mathsf{dtMPDA}}\xspace}

\begin{document}

\title{Analyzing Timed Systems Using Tree Automata}

\author{S. Akshay}
\address{Dept. of CSE, IIT Bombay, Powai, Mumbai 400076, India}
\email{akshayss@cse.iitb.ac.in}

\author{Paul Gastin}
\address{LSV, ENS-Cachan, CNRS, Universit\'e Paris-Saclay, 94235 Cachan, France}
\email{paul.gastin@lsv.ens-cachan.fr}      
  
\author{Shankara Narayanan Krishna}
\address{Dept. of CSE, IIT Bombay, Powai, Mumbai 400076, India}
\email{krishnas@cse.iitb.ac.in}

\thanks{The authors gratefully acknowledge support from DST-CEFIPRA projects AVeRTS and EQuaVe, UMI-ReLaX and DST-INSPIRE faculty award [IFA12-MA-17].}


\begin{abstract}
  Timed systems, such as timed automata, are usually analyzed using their
  operational semantics on timed words.  The classical region abstraction for
  timed automata reduces them to (untimed) finite state automata with the same
  time-abstract properties, such as state reachability.  We propose a new
  technique to analyze such timed systems using finite tree automata instead of
  finite word automata.  The main idea is to consider timed behaviors as graphs
  with matching edges capturing timing constraints.  When a family of graphs has
  bounded tree-width, they can be interpreted in trees and MSO-definable 
  properties of such graphs can be checked using tree automata.
  The technique is quite general and applies to many timed systems.
  In this paper, as an example, we develop the technique on timed pushdown
  systems, which have recently received considerable attention.  Further, we
  also demonstrate how we can use it on timed automata and timed multi-stack
  pushdown systems (with boundedness restrictions).
\end{abstract}

\maketitle

\section{Introduction}
\label{sec-intro}
The advent of timed automata \cite{AD94} marked the beginning of an era in  the
verification of real-time systems.  Today, timed automata form one of the well accepted real-time modelling formalisms, using real-valued variables called clocks to capture time constraints.  The decidability of the emptiness problem for timed automata is achieved using the notion of region abstraction.  This gives a sound and finite abstraction of an infinite state system, and has paved the way for state-of-the-art tools like UPPAAL~\cite{uppaal}, which have successfully been used in the verification of several complex timed systems. In recent times \cite{lics12,CL15,hscc15} there has been a lot of interest in the theory of verification of more complex timed systems enriched with features such as concurrency, communication between components and recursion with single or multiple threads.  In most of these approaches, decidability has been obtained by cleverly extending the fundamental idea of region or zone abstractions.

In this paper, we give a technique for analyzing timed systems in general, inspired from a completely different approach based on graphs and tree automata. This approach has been exploited for analyzing various types of \emph{untimed systems}, e.g., \cite{MP11,cyriac-phd}. The basic template of this approach has three steps: (1) capture the behaviors of the system as graphs, (2) show that the class of graphs that are actual behaviors of the system is MSO-definable, and (3) show that this class of graphs has bounded tree-width (or clique-width or split-width), or restrict the analysis to such bounded behaviors. Then, non-emptiness of the given system boils down to the satisfiability of an MSO sentence on graphs of bounded tree-width, which is decidable by Courcelle's theorem.  
But, by providing a direct construction of the tree automaton, it is possible to obtain a good complexity for the decision procedure.

We lift this technique to deal with timed systems by abstracting timed word behaviors of timed systems as graphs consisting of untimed words with additional time-constraint edges, called words with timing constraints (TCWs). The main complication here is that a TCW describes an abstract run of the timed system, where the constraints are recorded but not checked.  The TCW corresponds to an actual concrete run if and only if it is \emph{realizable}, i.e., we can find time-stamps realizing the TCW.  Thus, we are interested in the class of graphs which are \emph{realizable} TCWs. 

For this class of graphs, the above template tells us that we need to show  (i) these graphs have a bounded tree-width and (ii) the property of being a realizable TCW is MSO-definable. Then by Courcelle's theorem we obtain a tree automaton accepting this. However, as mentioned earlier, the MSO to tree-automaton approach does not give a good complexity in terms of size of the tree automaton. To obtain an optimal complexity, instead of going via Courcelle's theorem, we directly build a tree automaton.
Using tree decompositions of graph behaviours having bounded split/tree-width and constructing tree automata proved to be a very successful technique for the analysis of \emph{untimed} infinite state systems \cite{MP11,CGK12,cyriac-phd,CG14}.  This paper opens up this powerful technique for analysis of \emph{timed} systems. 

Thus our contributions are the following. We start by showing that behaviors of timed systems can be written as words with \emph{simple} timing constraints (\STCWs), i.e., words where each position has at most one timing constraint (incoming or outgoing) attached to it. This is done by breaking each transition of the timed system into a sequence of ``micro-transitions'', so that at each micro-transition, only one timing constraint is attached.

 \begin{wrapfigure}[10]{c}[.5\columnsep]{4.2cm}
  \includegraphics[scale=.8,page=1]{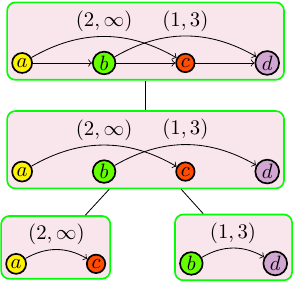}
\end{wrapfigure}
Next, we show that \STCWs that arise as behaviours of certain classes of timed systems (e.g., timed automata or timed pushdown systems) are graphs of bounded \emph{special} tree-width.  Special tree-width is a graph complexity measure, 
arising out of a special tree decomposition of a graph, as  
introduced by Courcelle in~\cite{courcelle10}. To establish the bound, 
we play a so-called \emph{split-game}, which gives a bound on what is called the \emph{split-width}, a notion that was introduced for graph behaviours of untimed systems and which has proven to be very useful for untimed systems~\cite{CGK12,cyriac-phd,CG14}. Establishing a relationship between the split-width and special tree-width, we obtain a bound on special tree-width (which also implies a bound on general tree-width as shown in~\cite{courcelle10}).  
As a result of this bound, we infer that our graphs admit  
binary tree decompositions  as depicted in the adjoining figure. Each node of the tree decomposition depicts a partial behaviour
 of the system in a  bounded manner.  By combining these behaviours as we go up the tree, we obtain a full behaviour of the system.

Our final and most technically challenging step is to construct tree automata
that works on such tree decompositions and accepts only those whose roots are
labeled by \emph{realizable} \STCWs generated by the timed system $\Sys$.  Thus,
checking non-emptiness of the timed system reduces to checking non-emptiness for
this tree automaton (which is PTIME in the size of the tree automaton).  The
construction of this tree automaton is done in two phases.  First, given the
bound on special tree-width of the graph behaviours of $\Sys$, we construct a
tree automaton that accepts all such trees whose roots are labeled by realizable
\STCWs with respect to the maximal constant given by the system, and whose nodes
are all bounded.  For this, we need to check that the graphs at the root are
indeed words with valid timing constraints and there exists a time-stamping that
realizes the \STCW. If we could maintain the partially constructed \STCW in a
state of the tree automaton along with a guess of the time-stamp at each vertex,
we could easily check this.  However, the tree automaton has only finitely many
states, so while processing the tree bottom-up, we need a finite abstraction of
the \STCW which remembers only finitely many time values.  We show that this is
indeed possible by coming up with an abstraction where it is sufficient to
remember the modulo $M$ values of time stamps, where $M$ is one more than the
maximal constant that appears in the timed system.

In the second phase, we refine this tree automaton to obtain another tree automaton that only accepts those trees that are generated by the system. Yet again, the difficulty is to ensure correct matching of (i) clock constraints between points   where a clock is reset and a constraint is checked, (ii) push and pop transitions by keeping   only finite amount of information in the states of the tree automaton.    Once both of these are done, the final tree automaton satisfies all our constraints.

To illustrate the technique, we have reproved the decidability of non-emptiness
of timed automata and timed pushdown automata (TPDA), by showing that both these
models have a split-width ($|X|+4$ and $4|X|+6$) that is linear in the number of
clocks, $|X|$, of the given timed system.  This bound directly tells us the
amount of information that we need to maintain in the construction of the tree
automata.  For TPDA we obtain an \textsc{ExpTime} algorithm, matching the known
lower-bound for the emptiness problem of TPDA~\cite{lics12}.  For timed automata,
since the split-trees are word-like (at each binary node, one subtree is small)
we may use word automata instead of tree automata, reducing the complexity from
\textsc{ExpTime} to \textsc{PSpace}, again matching the lower-bound~\cite{AD94}.
Interestingly, if one considers TPDA with no explicit clocks, but the stack is
timed, then the split-width is a constant, 2.  In this case, we have a
polynomial time procedure to decide emptiness, assuming a unary encoding of
constants in the system.  To further demonstrate the power of our technique, we
derive a new decidability result for non-emptiness of timed multi-stack pushdown
automata under bounded rounds, by showing that the split-width of this model is
again linear in the number of clocks, stacks and rounds.  Exploring decidable
subclasses of untimed multi-stack pushdown systems is a very active research
area \cite{A12,CHL12,latin10,TMP07,TNP14}, and our technique can extend these
to handle time.

It should be noticed that the tree automaton used to check emptiness of the
timed system is essentially the intersection of two tree automata.  The tree
automaton for validity/realizability (Section~\ref{sec:AKMV}), by far the most
involved construction, is independent of the timed system under study.  The
second tree automaton depends on the system and is rather easy to construct
(Section~\ref{sec:TA-sys}).  Hence, to apply the technique to other systems, one
only needs to prove the bound on split-width and to show that their runs can be
captured by tree automata.  This is a major difference compared to many existing techniques for timed systems which are highly system dependent. For instance, for the well-established models of TPDA, that we considered above, in~\cite{BER94,lics12} it is shown that the basic problem of checking emptiness is decidable (and $\mathsf{EXPTIME}$-complete) by re-adapting the technique of region abstraction each time (and possibly untiming the stack) to obtain an untimed pushdown automaton.  Finally, an orthogonal approach to deal with timed systems was developed in~\cite{CL15}, where the authors show the decidability of the non-emptiness problem for a class of timed pushdown automata by reasoning about sets with timed-atoms.

An extended abstract of this paper was presented in~\cite{concur16}. There are however, significant differences from that version as we detail now. First, the technique used to prove the main theorem of building a tree automaton to check realizability is completely different. In~\cite{concur16}, we showed an automaton which checks for non-existence of negative weight cycles (which implies the existence of a realizable time-stamping). This required a rather complicated proof to show that the constants can be bounded. In fact, we first build an infinite state tree automaton and then show that we can get a finite state abstraction for it. In contrast, our proof in this article directly builds a tree automaton that checks for existence of time-stamps realizing a run. This allows us to improve the complexity and also gives a less involved proof. 
Second, we complete the proof details and compute the complexity for tree automata for multistack pushdown systems with bounded round restriction, which was announced in \cite{concur16}. All sketches from the earlier version have been replaced and enhanced by rigorous proofs in this article. Further, several supporting examples and intuitive explanations have been added throughout to aid in understanding.

%%% Local Variables:
%%% mode: latex
%%% TeX-master: "main.tex"
%%% End:

\section{Graphs for behaviors of timed systems}
\label{sec-prelims}
We fix an alphabet $\Sigma$ and use $\Sigma_\varepsilon$ to denote $\Sigma\cup \{\varepsilon\}$ where $\varepsilon$ is the silent action.  For a non-negative integer $M$, we also fix $\cI(M)$, to be a finite set of closed intervals whose endpoints are integers between $0$ and $M$, and which contains the special interval $[0,0]$. When $M$ is irrelevant or clear from the context, we just write $\cI$ instead of $\cI(M)$. Further, for an interval $I$, we will sometimes use $I.up$ to denote its upper/right end-point and $I.lo$ to denote its lower/left end-point. For a set $S$, we use ${\leq}\subseteq {S\times S}$ to denote a partial or total order on $S$.  For any $x,y\in S$, we write $x<y$ if $x\leq y$ and $x\neq y$, and $x\lessdot y$ if $x<y$ and there does not exist $z\in S$ such that $x<z<y$.

\subsection{Preliminaries: Timed words and timed (pushdown) automata}
An \emph{$\varepsilon$-timed word} is a sequence $w=(a_1,t_1)\ldots(a_n,t_n)$
with $a_1\ldots a_n\in\Sigma_\varepsilon^+$ and $(t_i)_{1\leq i\leq n}$ is a
non-decreasing sequence of real time values.  If $a_i\neq\varepsilon$ for all $1\leq i\leq n$, then $w$ is a \emph{timed word}.  The projection on $\Sigma$ of an $\varepsilon$-timed word is the timed word obtained by removing $\varepsilon$-labelled positions.  We define the two basic system models that we consider in this article.

Dense-timed pushdown automata (TPDA), introduced in \cite{lics12}, are an extension of timed automata~\cite{AD94}, and operate on a finite set of real-valued clocks and a stack which holds symbols with their ages.  The age of a symbol in the stack represents time elapsed since it was pushed on to the stack. Formally, a TPDA $\Sys$ is a tuple $(S, s_{0}, \Sigma, \Gamma, \Delta, X, F)$ where $S$ is a finite set of states, $s_{0} \in S$ is the initial state, $\Sigma$, $\Gamma$, are respectively a finite set of input, stack symbols, $\Delta$ is a finite set of transitions, $X$ is a finite set of real-valued variables called clocks, $F\subseteq S$ are final states.  A transition $t \in \Delta$ is a tuple $(s, \gamma, a, \op, R, s')$ where $s, s' \in S$, $a\in \Sigma$, $\gamma$ is a finite conjunction of atomic formulae of the kind $x\in I$ for $x \in X$ and $I \in \mathcal{I}$, $R \subseteq X$ are the clocks reset, $\op$ is one of the following stack operations:
\begin{enumerate}
  \item \nop does not change the contents of the stack,
  \item $\push_c$ where $c \in \Gamma$ is a push operation that adds $c$ on top of the stack, with age 0.
  \item $\pop^I_c$ where $c \in \Gamma$ is a stack symbol and $I \in \mathcal{I}$ is an interval, is a pop operation 
  that removes the top most symbol of the stack if it is a $c$ with age in the interval $I$. 
\end{enumerate}
Timed automata (TA) can be seen as TPDA using \nop operations only.  This
definition of TPDA is equivalent to the one in~\cite{lics12}, but allows
checking conjunctive constraints and stack operations together.  In \cite{CL15},
it is shown that TPDA of~\cite{lics12} are expressively equivalent to timed
automata with an untimed stack.  Nevertheless, our technique is oblivious to
whether the stack is timed or not, hence we focus on the syntactically more
succinct model of TPDA with timed stacks and get good complexity bounds.

The operational semantics of TPDA and TA can be given in terms of timed words and we refer to~\cite{lics12} for the formal definition. Instead, we are interested in an alternate yet equivalent semantics for TPDA using graphs with timing constraints, that we define next.

\subsection{Abstractions of timed behaviors as graphs}
\label{tcwords}
\begin{definition}
  A \emph{word with timing constraints} (\TCW) over $\Sigma,\cI$ is a structure $\stcw=(V,\procrel,(\matchrel^I)_{I\in\mathcal{I}},\lambda)$
where $V$ is a finite set of positions (also refered to as vertices or points), $\lambda\colon V\to \Sigma_\varepsilon$ labels each position,  the reflexive transitive closure ${\leq}={\procrel}^*$ is a total order on $V$ and ${\procrel}={\lessdot}$ is the successor relation,  ${\matchrel^I}\subseteq{<}={\procrel}^+$ gives the pairs of positions carrying a timing constraint associated with the interval $I$.
\end{definition} 
For any position $i \in V$, the \emph{indegree} (resp.\ \emph{outdegree}) of $i$
is the number of positions $j$ such that $(j,i) \in \matchrel$ (resp.\ $(i,j)
\in \matchrel$). A \TCW is \emph{simple} (denoted \STCW) if each position has at most one timing constraint
(incoming or outgoing) attached to it, i.e., for all $i\in V$,
\emph{indegree($i$) + outdegree($i$)} $\leq$ 1.  A
\TCW is depicted below with positions $1,2, \dots, 5$ labelled over $\{a,b\}$.
\emph{indegree}(4)=1, \emph{outdegree}(1)=1 and \emph{indegree}(3)=0.  The curved edges decorated with intervals connect the positions related by $\matchrel$,
while straight edges are the successor relation $\procrel$.  Note that this \TCW is
\emph{simple}.

\hfil
\includegraphics[scale=.6,page=2]{tikz-pics}
\hfil

Consider a \TCW $\stcw=(V,\procrel,(\matchrel^I)_{I\in\mathcal{I}},\lambda)$ with
$V=\{1,\ldots,n\}$.  A timed word $w$ is a \emph{realization} of $\stcw$ if it is the projection on $\Sigma$ of an $\varepsilon$-timed word
$w'=(\lambda(1),t_1)\ldots(\lambda(n),t_n)$ such that $t_j-t_i\in I$
for all $(i,j)\in{\matchrel^I}$. In other words, a \TCW is realizable 
if there exists a timed word $w$ which is a realization of $\stcw$. For example, the timed word $(a,0.9)(b,2.1)(a,2.1)(b,3.9)(b,5)$ is a realization of the \TCW depicted  above, while $(a,1.2)(b,2.1)$ $(a,2.1)(b,3.9)(b,5)$ is not. %For a given set of \TCWs $S$, we use $\Real(S)$ to denote the set of all timed words that are realizations of \TCWs from $S$.

\subsection{Semantics of TPDA (and TA) as simple \TCWs}\label{sec:tpda-sem}
We define the semantics in terms of simple \TCWs.  An \STCW
$\stcw=(V,\procrel,(\matchrel^I)_{I\in\mathcal{I}},\lambda)$
 is said to be generated or
accepted by a TPDA $\Sys$ if there is an accepting abstract run
$\rho=(s_0,\gamma_1,a_1,\op_1,R_1,s_1)$ $(s_1,\gamma_2,a_2,\op_2,R_2,s_2)\cdots$
$(s_{n-1},\gamma_n,a_n,\op_n,R_n,s_n)$ of $\Sys$ such that, $s_n\in F$ and
\begin{itemize}
  \item the sequence of push-pop operations is well-nested: in each prefix
  $\op_1\cdots\op_k$ with $1\leq k\leq n$, number of pops is at most number of pushes, and in the full sequence $\op_1\cdots\op_n$, they are equal.
  
  \item We have $V=V_0\uplus V_1\uplus \cdots \uplus V_n$ with $V_i\times 
  V_j\subseteq{\procrel}^+$ for $0\leq i<j\leq n$.
  Each transition $\delta_i=(s_{i-1},\gamma_i,a_i,\op_i,R_i,s_i)$ gives rise to
  a sequence of consecutive points $V_i$ in the \STCW. The transition $\delta_i$
  is simulated by a sequence of ``micro-transitions'' as depicted below (left)
  and it represents an \STCW shown below (right).  Incoming red edges
  check guards from $\gamma_i$ (wrt different clocks) while outgoing green edges
  depict resets from $R_i$ that will be checked later.  
  Further, the outgoing edge on the central node labeled $a_i$ represents a push
  operation on stack.  Finally, the blue $[0,0]$ labeled edge denotes a $0$-delay constraint as explained below.

  \noindent\smallskip 
  \includegraphics[width=85mm,page=2]{gpicture-pics}
  \hfill
  \includegraphics[scale=.8,page=4]{tikz-pics}
 
  \noindent\smallskip 
  where $\gamma_i=\gamma_i^1\wedge\cdots\wedge\gamma_i^{h_i}$ are conjunctions of atomic clock constraints and
  $R_i=\{x_1,\ldots,x_m\}$ are resets.  The first and last micro-transitions, corresponding
  to the reset of a new clock $\zeta$ and checking of constraint $\zeta=0$ ensure
  that all micro-transitions in the sequence occur simultaneously.  
  We have a point in $V_i$ for each micro-transition (excluding the
  $\varepsilon$-micro-transitions between $\delta_i^{x_j}$).  Hence, $V_i$
  consists of a sequence
 $ \ell_i \procrel \ell_i^{1}  \procrel \cdots \procrel \ell_i^{h_i}
  \procrel p_i \procrel 
  r_i^{1} \procrel \cdots \procrel r_i^{g_i} \procrel r_i $
  where $g_i$ is the number of timing constraints corresponding to clocks 
  reset during transition $i$ and checked afterwards. 
  Thus, the reset-loop on a clock is fired as many times (0 or more) as a
  constraint is checked on this clock until its next reset.  This ensures that
  the \STCW remains \emph{simple}.
  Similarly, $h_i$ 
 is the number of timing constraints checked in $\gamma_i$. 
  We have $\lambda(p_i)=a_i$ and all other points are labelled $\varepsilon$.
  The set $V_0$ encodes the initial resets of clocks that will be checked 
  before being reset.  So we let $R_0=X$ and $V_0$ is 
 $ \ell_0 \procrel
  r_0^{1} \procrel \dots \procrel r_0^{g_0} \procrel r_0 \,.
 $ 

  \item for each $I\in \cI$, the relation for timing constraints can be partitioned as ${\matchrel^I}={\matchrel^{s\in I}}\uplus\biguplus_{x\in X\cup\{\zeta\}}{\matchrel^{x\in I}}$ where
  \begin{itemize}
    \item ${\matchrel}^{\zeta\in I}=\{(\ell_i,r_i)\mid 0\leq i\leq n\}$ and $I=[0,0]$. 
  
    \item We have $p_i\matchrel^{s\in I} p_j$ if $\op_i={\push}_b$ is a push and
    $\op_j={\pop}_b^{I}$ is the matching pop (same number of pushes and pops in 
    $\op_{i+1}\cdots\op_{j-1}$).
    \item for each $0\leq i<j\leq n$ such that the $t$-th conjunct of 
    $\gamma_j$ is $x\in I$ and $x\in R_i$ and $x\notin R_k$ for $i<k<j$, we 
    have $r_i^s\matchrel^{x\in I} \ell_j^t$ for some $1\leq s\leq g_i$. 
    Therefore, every point $\ell_i^t$ with $1\leq
    t\leq h_i$ is the target of a timing constraint.
    Moreover, every reset point $r_i^s$ for $1\leq s\leq g_i$ should be the 
    source of a timing constraint. That is, denoting $\matchrel^x=\bigcup_{I'\in \cI} \matchrel^{x\in I'}$, we must have $r_i^s\in\dom({\matchrel}^x)$ for some $x\in 
    R_i$.  Also, for each $i$, the reset points $r_i^{1},\ldots,r_i^{g_i}$ are
    grouped by clocks (as suggested by the sequence of 
    micro-transitions simulating $\delta_i$): if $1\leq s<u<t\leq g_i$ and 
    $r_i^s,r_i^t\in\dom({\matchrel}^x)$ for some $x\in R_i$ then 
    $r_i^u\in\dom({\matchrel}^x)$.
        Finally, for each clock, we require that the timing constraints are
    well-nested: for all $u\matchrel^x v$ and $u'\matchrel^x v'$, 
    with $u, u' \in V_i$, 
    if $u<u'$ then $u'<v'<v$.
  \end{itemize}
\end{itemize}
We denote by $\STCW(\Sys)$ the set of simple \TCWs generated by $\Sys$ and define the language of $\Sys$ as the set of timed words that are \emph{realizations} of \STCWs generated by $\Sys$, i.e., $\cL(\Sys)=\{w\mid w\text{ is a realization of } \stcw\in \STCW(\Sys)\}$.  Indeed, this is equivalent to defining the language as the set of timed words accepted by $\Sys$, according to a usual operational semantics \cite{lics12}.  

The $\STCW$ semantics of timed automata (TA) can be obtained from the above
discussion by just ignoring the stack components (using \nop operations only).
To illustrate these ideas, a simple example of a timed automaton and an $\STCW$ that is generated by it is shown in Figure~\ref{ex-ta}.

\begin{figure}[h]
  \hfil
  \includegraphics[scale=.8,page=5]{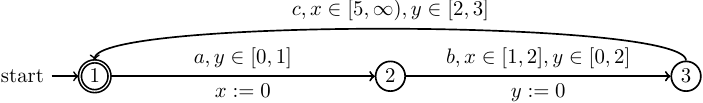}

\vspace*{7mm}
  
  \hfil
  \includegraphics[scale=1,page=6]{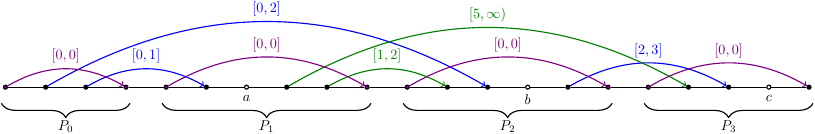}
  
  \caption{A timed automaton (top) with 2 clocks $x,y$.  An $\STCW$ generated by
  an accepting run of the TA is depicted just below.  The blue edges represent
  matching relations induced by clock $y$, while the green represent those
  induced by clock $x$.  The violet edges are the $[0,0]$ timing constraints for
  the extra clock $\zeta$ ensuring that all points in some $V_i$ representing
  transition $i$ occur precisely at the same time.  Black lines are process
  edges.}
  \label{ex-ta}
\end{figure}

\noindent We now identify some important properties satisfied by \STCWs
generated from a TPDA. Let $\stcw=(V,\procrel,(\matchrel^I)_{I\in\mathcal{I}},\lambda)$
be a
\STCW.
We say that \stcw is \emph{well timed} w.r.t.\ a set of clocks $Y$ and a stack
$s$ if for each interval $I\in \cI$, the $\matchrel^I$ relation can be partitioned as
${\matchrel^I}={\matchrel^{s\in I}}\uplus\biguplus_{x\in Y}{\matchrel^{x\in I}}$ where
\begin{enumerate}[label=$(\mathsf{T}_{\arabic*})$,ref=$\mathsf{T}_{\arabic*}$]
  \item\label{item:T1}
  the stack relation $\matchrel^s=\bigcup_{I\in \cI} \matchrel^{s\in I}$ corresponds to the matching push-pop events, 
  hence it is well-nested: for all $i\matchrel^s j$ and $i'\matchrel^s j'$, if 
  $i<i'<j$ then $i'<j'<j$.
  \item\label{item:T2} For each clock $x\in Y$, the relation $\matchrel^x=\bigcup_{I\in \cI} \matchrel^{x\in I}$ corresponds to the timing constraints for clock $x$ and is well-nested: for all $i\matchrel^x j$ and $i'\matchrel^x j'$, if $i<i'$ are in the same  \emph{$x$-reset block} (i.e., a maximal consecutive sequence
  $i_1\lessdot\cdots\lessdot i_n$ of positions in the domain of $\matchrel^x$),
  and $i<i'<j$, then $i'<j'<j$.  Each guard should be matched with the closest
  reset block on its left: for all $i\matchrel^x j$ and $i'\matchrel^x j'$, if
  $i<i'$ are not in the same $x$-reset block then $j<i'$  (see Figure~\ref{fig:well-timed}).
\end{enumerate}

\begin{figure}[tbp]
  \centering
  \includegraphics[width=8.5cm,page=1]{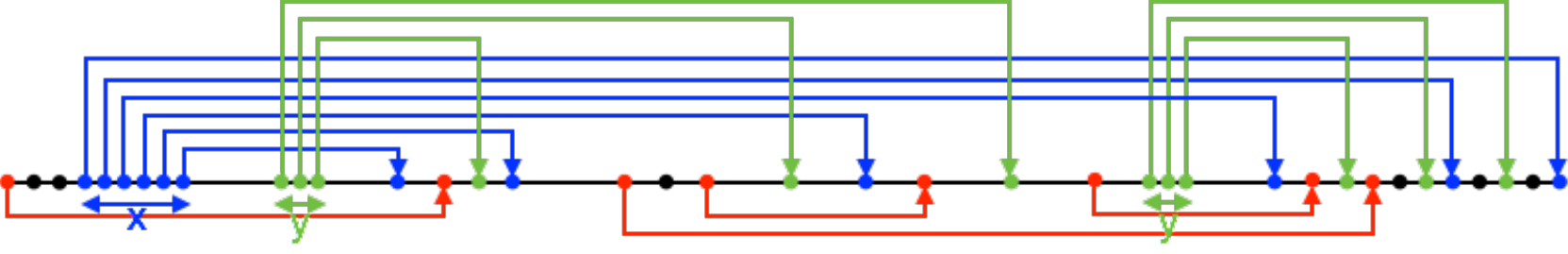}
  \caption{Well-timed: relations $\matchrel^s$ in red, $\matchrel^x$ in blue 
  and  $\matchrel^y$ in green.}
  \label{fig:well-timed}
\end{figure}

It is then easy to check that \STCWs defined by a TPDA with set of clocks $X$ are well-timed for the set of clocks $Y=X\cup\{\zeta\}$, i.e., satisfy the properties above: (note that we obtain the same for TA by just ignoring the stack edges, i.e., (\ref{item:T1}) above)

\begin{lemma}\label{app:well-timed}
Simple \TCWs defined by a TPDA with set of clocks $X$ are well-timed wrt.\ set 
of clock $Y=X\cup\{\zeta\}$, i.e., satisfy properties (T1) and (T2).
\end{lemma} 
\begin{proof}
The first condition \eqref{item:T1} is satisfied by $\STCW(\Sys)$ by definition.  For \eqref{item:T2}, let $i\matchrel^x j$ and $i'\matchrel^x j'$ for some clock $x \in X$.  If $i, i'$ are points in the same $x$-reset block for some $x \in X$, then by construction of $\STCW(\Sys)$, if $i<i'$ then $i' < j' < j$ which gives well nesting.  Similarly, if  $i<i'$ are points in different $x$-reset blocks, then by definition of  $\STCW(\Sys)$, we have $j<i'$. Also, it is clear that the new clock $\zeta$  satisfies \eqref{item:T2}.
\end{proof}

%%% Local Variables:
%%% mode: latex
%%% TeX-master: "main.tex"
%%% End:

\section{Bounding the width of graph behaviors of timed systems}
\label{sec:stw-tpda}
In this section, we check if the graphs (STCWs) introduced in the previous section have a bounded tree-width. As a first step towards that, we introduce \emph{special tree terms} (\STTs) from Courcelle~\cite{courcelle10} and their semantics as labeled graphs. 

\subsection{Preliminaries: special tree terms and tree-width}
It is known \cite{courcelle10} that special tree terms using at most $K$ colors (\kSTTs) define graphs of
``special'' tree-width at most $K-1$. Formally, a $(\Sigma,\Gamma)$-labeled graph is a tuple $G=(V_G,(E_\gamma)_{\gamma\in\Gamma},\lambda)$ where $\lambda\colon V_G\to\Sigma$ is the vertex labeling and $E_\gamma \subseteq V_G^2$ is the set of edges for each label $\gamma\in\Gamma$. Special tree terms form an algebra to define labeled graphs. The syntax of \kSTTs over $(\Sigma,\Gamma)$ is given by
\[
\stt ::= (i,a) \mid \add{i}{j}{\gamma} \stt \mid \forget{i} \stt \mid
\rename{i}{j} \stt \mid \stt \sttunion \stt\]

where $a \in \Sigma$, $\gamma \in \Gamma$ and $i,j\in[K]=\{1,\ldots,K\}$ are
colors.   The semantics of each \kSTT is a colored graph $\sem\stt=(G_\stt,\chi_\stt)$ where $G_\stt$ is a $(\Sigma,\Gamma)$-labeled graph and $\chi_\stt\colon [K]\to V_G$ is a partial injective function assigning a vertex of $G_\stt$ to some colors.

\begin{itemize}
  \item $\sem{(i,a)}$ consists of a single $a$-labeled vertex with color $i$.

  \item $\add{i}{j}{\gamma}$ adds a $\gamma$-labeled edge to the vertices
  colored $i$ and $j$ (if such vertices exist).
  
  Formally, if $\sem{\stt}=(V_G,(E_\gamma)_{\gamma\in\Gamma},\lambda,\chi)$ then
  $\sem{\add{i}{j}{\alpha}\stt}=(V_G,(E'_\gamma)_{\gamma\in\Gamma},\lambda,\chi)$ 
  with $E'_\gamma=E_\gamma$ if $\gamma\neq\alpha$ and
  $E'_\alpha= \begin{cases}
    E_\alpha & \text{if } \{i,j\}\not\subseteq\dom(\chi) \\
    E_\alpha\cup\{(\chi(i),\chi(j))\} & \text{otherwise.}
  \end{cases}$
  
  \item $\forget{i}$ removes color $i$ from the domain of the color map.
  
  Formally, if $\sem{\stt}=(V_G,(E_\gamma)_{\gamma\in\Gamma},\lambda,\chi)$ then
  $\sem{\forget{i}\stt}=(V_G,(E_\gamma)_{\gamma\in\Gamma},\lambda,\chi')$ 
  with $\dom(\chi')=\dom(\chi)\setminus\{i\}$ and $\chi'(j)=\chi(j)$ for all 
  $j\in\dom(\chi')$.
  
    \item $\rename{i}{j}$ exchanges the colors $i$ and $j$.
  
  Formally, if $\sem{\stt}=(V_G,(E_\gamma)_{\gamma\in\Gamma},\lambda,\chi)$ then
  $\sem{\rename{i}{j}\stt}=(V_G,(E_\gamma)_{\gamma\in\Gamma},\lambda,\chi')$ 
  with $\chi'(\ell)=\chi(\ell)$ if $\ell\in\dom(\chi)\setminus\{i,j\}$, 
  $\chi'(i)=\chi(j)$ if $j\in\dom(\chi)$ and $\chi'(j)=\chi(i)$ if 
  $i\in\dom(\chi)$.
  
  \item Finally, $\sttunion$ constructs the disjoint union of the two graphs
  provided they use different colors.  This operation is undefined otherwise.
  
  Formally, if $\sem{\stt_i}=(G_i,\chi_i)$ for $i=1,2$ and 
  $\dom(\chi_1)\cap\dom(\chi_2)=\emptyset$ then 
  $\sem{\stt_1\sttunion\stt_2}=(G_1\uplus G_2,\chi_1\uplus\chi_2)$.
  Otherwise, $\stt_1\sttunion\stt_2$ is not a valid \STT.
\end{itemize}

The \emph{special tree-width} of a graph $G$ is defined as the least $K$ such that $G=G_\stt$ for some $(K+1)$-\STT $\stt$. See~\cite{courcelle10} for more details and its relation to tree-width. For \TCWs, we have successor edges and
$\matchrel^I$-edges carrying timing constraints, so we take
$\Gamma=\{\procrel\}\cup\{(x,y)\mid x\in\N, y\in\overline{\N}\}$ with
$\overline{\N}=\N\cup\{\infty\}$. In this paper, we will actually make use of \STTs with the following restricted syntax, which are sufficient and make our proofs simpler:
\begin{align*}
  \atomicSTT &::= (1,a) \mid \add{1}{2}{x,y}((1,a)\oplus(2,b)) \\
  \stt &::= \atomicSTT \mid \add{i}{j}{\procrel} \stt \mid \forget{i} \stt \mid
  \rename{i}{j} \stt \mid \stt \sttunion \stt 
\end{align*}
with $a,b\in\Sigma_\varepsilon$, $0\leq x<M$, $0\leq y<M$ or $y=+\infty$ for some $M \in \mathbb{N}$ and
$i,j\in[2K]=\{1,\ldots,2K\}$.  The terms defined by this grammar are called $(K,M)$-\STTs. Here, timing constraints are added directly between leaves in atomic \STTs which are then combined using disjoint unions and adding successor edges. 
For instance, consider  the 4-\STT given below 
$$\stt= \forget{3} \add{1}{3}{\procrel} \forget{2} \add{2}{4}{\procrel} \add{3}{2}{\procrel}
(\add{1}{2}{2,\infty}((1,a)\sttunion(2,c))  \sttunion
\add{3}{4}{1,3}((3,b)\sttunion(4,d)))$$

\noindent where $\add{i}{j}{\gamma}((i,\alpha)\sttunion(j,\beta))$,
$i,j\in\mathbb{N}, \alpha,\beta\in \Sigma_\varepsilon$ stands for $\rename{1}{i}
\rename{2}{j} \add{1}{2}{\gamma}((1,\alpha)\sttunion(2,\beta)).$
The term $\stt$ is depicted as a binary tree on the left of
Figure~\ref{fig:stcw} and its semantics $\sem{\stt}$ is the \STCW depicted at
the root of the tree in right of Figure~\ref{fig:stcw}, where only endpoints
labelled $a$ and $d$ are colored, as the other two colors were ``forgotten'' by
$\stt$.  The entire ``split-tree'' is depicted in the right of
Figure~\ref{fig:stcw} as will be explained in the next sub-section.  Abusing
notation, we will also use $\sem{\stt}$ for the graph $G_\tau$ ignoring the
coloring $\chi_t$.

\begin{figure}[h]
  \begin{center}
    \includegraphics[scale=1.1,page=15]{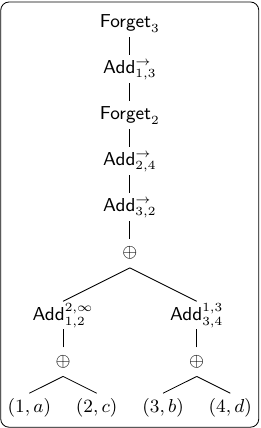}
    \hfil
    \includegraphics[scale=0.7,page=16]{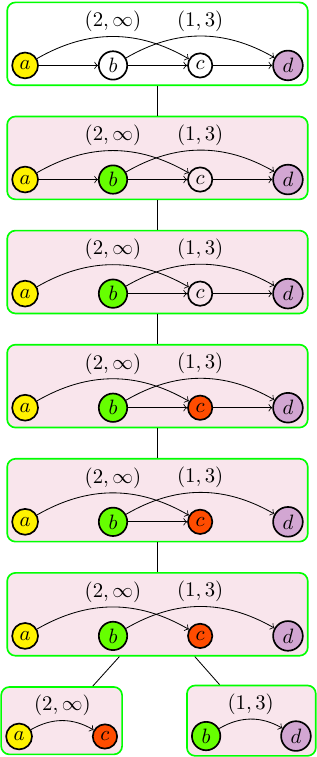}
  \end{center}
  \caption{An \STT and a simple $\TCW$ with its split-tree}
  \label{fig:stcw}
\end{figure}

\subsection{Split-TCWs and split-game}
\label{sec:split-game}
There are many ways to show that a graph (in our case, a simple \TCW) has a bounded special tree-width. We find it convenient to prove that a simple \TCW has bounded special tree-width by playing a split-game, whose game positions are simple \TCWs in which some successor edges have been cut, i.e., are missing. This approach has been used for graph behaviors generated by untimed pushdown systems in~\cite{CGK12}, but here we wish to apply it to reason about graph behaviors of timed systems. Formally, a \emph{split-\TCW} is a structure $(V,\procrel\cup{\hole},(\matchrel^I)_{I\in\mathcal{I}},\lambda)$ where $\procrel$ and
 $\hole$ are the present and absent successor edges (also called \emph{holes}), respectively, such that ${\procrel}\cap{\hole}=\emptyset$ and 
 $(V,\procrel\cup{\hole},(\matchrel^I)_{I\in\mathcal{I}},\lambda)$
  is a \TCW. Notice that, for a split-\TCW, ${\lessdot}={\procrel}\cup{\hole}$ and ${<}={\lessdot}^+$.  A \emph{block} or \emph{factor} of a split-\TCW is a maximal set of points of $V$ connected by $\procrel$.
We denote by $\EP(\stcw)\subseteq V$ the set of left and right endpoints of blocks of \stcw.  A left endpoint $e$ is one for which there is no $f$ with $f\procrel e$.  Right endpoints are defined similarly.  Points in
$V\setminus\EP(\stcw)$ are called internal. If $i$ is any point, $\mathsf{Block}(i)$
 denotes the unique block which contains $i$. 
 The number of blocks is the \emph{width} of $\stcw$: $\width(\stcw)=1+|{\hole}|$.  \TCWs may be identified with split-\TCWs of width 1, i.e., with ${\hole}=\emptyset$.  
A split-\TCW is \emph{atomic} if it consists
of a single point ($|V|=1$) or a single timing constraint with a hole
($V=\{p_1,p_2\}$, $p_1\hole p_2$, $p_1\matchrel^I p_2$). In what follows, we will use the notation $\matchrel$ for $\bigcup_{I\in \cI} \matchrel^I$ when convenient.

\emph{The split-game} is a two player turn based game
$\mathcal{G}=(V_\exists\uplus V_\forall,E)$ where Eve's set of game positions
$V_\exists$ consists of all connected (wrt.\ ${\procrel}\cup{\matchrel}$)
split-\TCWs and Adam's set of game positions $V_\forall$ consists of all disconnected (wrt.\ ${\procrel}\cup{\matchrel}$) split-\TCWs.  The edges $E$ of $\mathcal{G}$ reflect the moves of the players.
Eve's moves consist of splitting a block in two, i.e., removing one successor $(\procrel)$
edge in the graph.  Adam's moves amount to choosing a connected component (wrt.\ ${\procrel}\cup{\matchrel}$) of the
split-\TCW. Atomic split-\TCWs are terminal positions in the game: neither Eve
nor Adam can move from an atomic split-\TCW.
A play on a split-\TCW $\stcw$ is a path in $\mathcal{G}$ starting from
$\stcw$ and leading to an atomic split-\TCW. The cost of the play is
the maximum width of any split-\TCW encountered in the path.  Eve's objective is
to minimize the cost, while Adam's objective is to maximize it.

A strategy for Eve from a split-\TCW $\stcw$ can be described with a
\emph{split-tree} $T$ which is a binary tree labeled with split-\TCWs 
satisfying:
\begin{enumerate}
\item The root is labeled by $\stcw=\labelroot(T)$.
\item Leaves are labeled by atomic split-\TCWs.
\item Eve's move: Each unary node is labeled with some connected (wrt.\
  ${\procrel}\cup{\matchrel}$  ) split-\TCW $\stcw$ and its child is labeled with
some $\stcw'$ obtained by splitting a block of $\stcw$ in two, i.e., by
removing one successor $(\procrel)$ edge.  Thus, $\width(\stcw') = 1 + \width(\stcw)$.
\item Adam's move: Each binary node is labeled with some disconnected (wrt.\ ${\procrel}\cup{\matchrel}$)
split-\TCW $\stcw = \stcw_1 \uplus \stcw_2$ where $\stcw_1$ and $\stcw_2$ are the
labels of its children.  Note that $\width(\stcw)=\width(\stcw_1)+\width(\stcw_2)$.
\end{enumerate}

The \emph{width} of a split-tree $T$, denoted $\width(T)$, is the maximum width
of the split-\TCWs labeling the nodes of $T$.  In other words, the cost of the 
strategy encoded by $T$ is $\width(T)$ and this is the maximal cost of the plays starting from \stcw and following this strategy. A $K$-split-tree is a split-tree of width at most $K$.  

The \emph{split-width} of a simple (split-)\TCW $\stcw$ is the minimal cost of Eve's (positional) strategies starting from \stcw. In other words it is the minimum width of any split-tree for the simple \TCW. Notice that Eve has a strategy to decompose a \TCW $\stcw$ into \emph{atomic} split-\TCWs if and only if $\stcw$ is \emph{simple}, i.e, at most one timing constraint is attached to each point.  
An example of a split-tree is given in Figure~\ref{fig:stcw} (right).
Observe that the width of the split-tree is 4.  Hence the split-width of the
simple \TCW labeling the root is at most four.

Let $\KTCW$ (resp. $\KMTCW$) denote the set of simple \TCWs with split-width bounded by $K$ (resp. and using constants at most $M$) over the fixed alphabet $\Sigma$.  The crucial link between special tree-width and split-width is given by the following lemma.

\begin{lemma}\label{lem:split-tcw-stt}
  A (split) \STCW of split-width at most $K$ has special tree-width at
  most $2K-1$. 
\end{lemma}
Intuitively, we only need to keep colors for end-points of blocks.  Hence, each
block of an \STCW $\stcw$ needs at most two colors and if the width of $\stcw$
is at most $K$ then we need at most $2K$ colors. From this it can be
shown that a strategy of Eve of cost at most $K$ can be encoded by a $2K$-\STT,
which gives a special tree-width of at most $2K-1$.

\begin{proof} 
Let $\stc$ be a split-\TCW. Recall that we denote by $\EP(\stcw)$ the subset of
events that are endpoints of blocks of \stcw.  A left endpoint is an event
$e\in\Events$ such that there are no $f$ with $f\procrel e$.  We define
similarly right endpoints.  Note that an event may be both a left and right
endpoint.  The number of endpoints is at most twice the number of blocks:
$|\EP(\stcw)|\leq2\cdot\width(\stcw)$.

  We associate with every split-tree $T$ of width at most $K$ a $2K$-\STT
  $\overline{T}$ such that $\sem{\overline{T}}=(\stcw,\chi)$ where
  $\stcw=\labelroot(T)$ is the label of the root of $T$ and the range of $\chi$
  is the set of endpoints of $\stcw$: $\image(\chi)=\EP(\stcw)$.  Notice that
  $\dom(\chi)\subseteq[2K]$ since $\overline{T}$ is a $2K$-\STT. The
  construction is by induction on $T$.

  Assume that $\labelroot(T)$ is atomic. Then it is either an internal event labeled 
  $a\in\Sigma$, and we let $\overline{T}=(1,a)$.
  Or, it is a pair of events $e\matchrel^I f$ with a timing constraint
  $I=[c,d]$ and we let
  $\overline{T}=\add{1}{2}{c,d}((1,\lambda(e))\sttunion(2,\lambda(f)))$.
  
  If the root of $T$ is a binary node and the left and right subtrees are $T_1$
  and $T_2$ then $\labelroot(T)=\labelroot(T_1)\uplus\labelroot(T_2)$.  By
  induction, for $i=1,2$ the \STT $\overline{T}_i$ is already defined and we have
  $\sem{\overline{T}_i}=(\labelroot(T_i),\chi_i)$.  We first rename colors that
  are active in both \STTs.  To this end, we choose an injective map
  $f\colon\dom(\chi_1)\cap\dom(\chi_2)\to[2K]\setminus(\dom(\chi_1)\cup\dom(\chi_2))$.
  This is possible since $|\dom(\chi_i)|=|\image(\chi_i)|=|\EP(\labelroot(T_i))|$.
  Hence, $|\dom(\chi_1)|+|\dom(\chi_2)|=|\EP(\labelroot(T))|\leq 2K$.
  
  Assuming that $\dom(f)=\{i_1,\ldots,i_m\}$, we define
  $$
  \overline{T}=\overline{T}_1\sttunion
  \rename{i_1}{f(i_1)} \cdots \rename{i_m}{f(i_m)} \overline{T}_2 \,.
  $$
  
  Finally, assume that the root of $T$ is a unary node with subtree $T'$.  Then,
  $\labelroot(T')$ is obtained from $\labelroot(T)$ by splitting one block, i.e.,
  removing one word edge, say $e\procrel f$.  We deduce that $e$ and $f$ are
  endpoints of $\labelroot(T')$, respectively right and left endpoints.  By
  induction, the \STT $\overline{T'}$ is already defined.  We have
  $\sem{\overline{T'}}=(\labelroot(T'),\chi')$ and $e,f\in\image(\chi')$.  So let
  $i,j$ be such that $\chi'(i)=e$ and $\chi'(j)=f$.  We add the process edge with
  $\tau=\add{i}{j}\procrel \overline{T'}$.  Then we forget color $i$ if $e$ is no more
  an endpoint, and we forget $j$ if $f$ is no more an endpoint:
  \begin{align*}
    \tau' &=
    \begin{cases}
      \tau & \text{if $e$ is still an endpoint,}  \\
      \forget{i} \tau & \text{otherwise} 
    \end{cases}
    &
    \overline{T} &=
    \begin{cases}
      \tau' & \text{if $f$ is still an endpoint,}  \\
      \forget{j} \tau' & \text{otherwise.} 
    \end{cases}\tag*{\raisebox{-.65em}\qEd}
  \end{align*}
  \def\popQED{}
\end{proof}

\subsection{Split-width for timed systems} 
Viewing terms as trees, our goal in the next section will be to construct tree automata to recognize sets of $(K,M)$-\STTs, and thus capture the ($K$ split-width) bounded behaviors of a given system.  A possible way to show that these capture \emph{all} behaviors of the given system, is to show that we can find a $K$ such that all the (graph) behaviors of the given system have a $K$-bounded split-width. 

We do this now for a TPDA and also mention how to modify the proof for a timed automaton. In Section \ref{sec:mpda}, we also show how it extends to multi-pushdown systems. In fact, we prove a slightly more general result, by showing that all well-timed split-\STCWs defined in Section~\ref{sec:tpda-sem} for the set of clocks $Y=X\cup\{\zeta\}$ have bounded split-width (lifting the definition of well-timed to split-\STCWs).  From Lemma~\ref{app:well-timed}, it follows that the \STCWs defined by a TPDA with set of clocks $X$ are well-timed for the set of clocks $Y=X\cup\{\zeta\}$ and hence we obtain a bound on the split-width as required. Let $\STCW(\Sys)$ represent the \STCWs that are 
defined by a TPDA \Sys.

\begin{lemma}\label{TPDA-bound}
  The split-width of a \STCW which is well-timed w.r.t.\ a set of clock $Y$ is
  bounded by $4|Y|+2$.
\end{lemma}

\begin{proof}
This lemma is proved by playing the split-width game between \emph{Adam} and
\emph{Eve}.  \emph{Eve} should have a strategy to disconnect the word without
introducing more than $4|Y|+2$ blocks.  The strategy of \emph{Eve} uses three
operations processesing the word from right to left.

\vspace*{-1ex}
\subparagraph{\bf{Removing an internal point}.} If the last/right-most event on the
word (say event $j$) is not the target of a $\matchrel$ relation, then she will
split the $\procrel$-edge before the last point, i.e., the edge between point
$j$ and its predecessor.

\vspace*{-1ex}
\subparagraph{\bf{Removing a clock constraint}}
\label{app:rem-clk}
Assume that we have a timing constraint $i\matchrel^x j$ where $j$ is the last
point of the split-\TCW.  Then, by \eqref{item:T2} we deduce that $i$ is
\emph{the first} point of the \emph{last reset block} for clock $x$. 
\emph{Eve} splits three $\procrel$-edges to detach the matching pair 
$i\matchrel^x j$: these three edges are those connected to $i$ and $j$.  Since
the matching pair $i\matchrel^x j$ is atomic, \emph{Adam} should continue the 
game from the
remaining split-\TCW $\stcw'$.  Notice that we have now a hole instead of
position $i$.  We call this a reset-hole for clock $x$.  For instance, starting
from the split-\TCW of Figure~\ref{fig:well-timed}, and removing the last timing
constraint of clock $x$, we get the split-\TCW on top of
Figure~\ref{fig:reset-hole}.  Notice the reset hole at the beginning of the $x$
reset block.

During the inductive process, we may have at most one such reset hole for each 
clock $x\in Y$. Note that the first time we remove the last point which is a timing constraint for clock $x$,
we create a hole in the last reset block of $x$ which contains a sequence of reset points for $x$,  by removing two edges. 
This hole is created by removing the  leftmost point in the reset block. As we keep removing points from the right which are timing constraints 
for $x$, this hole widens in the reset block, by removing each time, just one edge in the reset block.   
Continuing the example, if we detach the last internal point and then the 
timing constraint of clock $y$ we get the split-\TCW in the middle of 
Figure~\ref{fig:reset-hole}. Now, we have one reset hole for clock $x$ and one 
reset hole for clock $y$.

We continue by removing from the right, one internal point, one timing
constraint for clock $x$, one timing constraint for clock $y$, and another
internal point, we get the split-\TCW at the bottom of
Figure~\ref{fig:reset-hole}.  Notice that we still have a \emph{single} reset
hole for each clock $x,y$.

\begin{figure}[tbp]
  \centering
  \includegraphics[width=8.5cm,page=2]{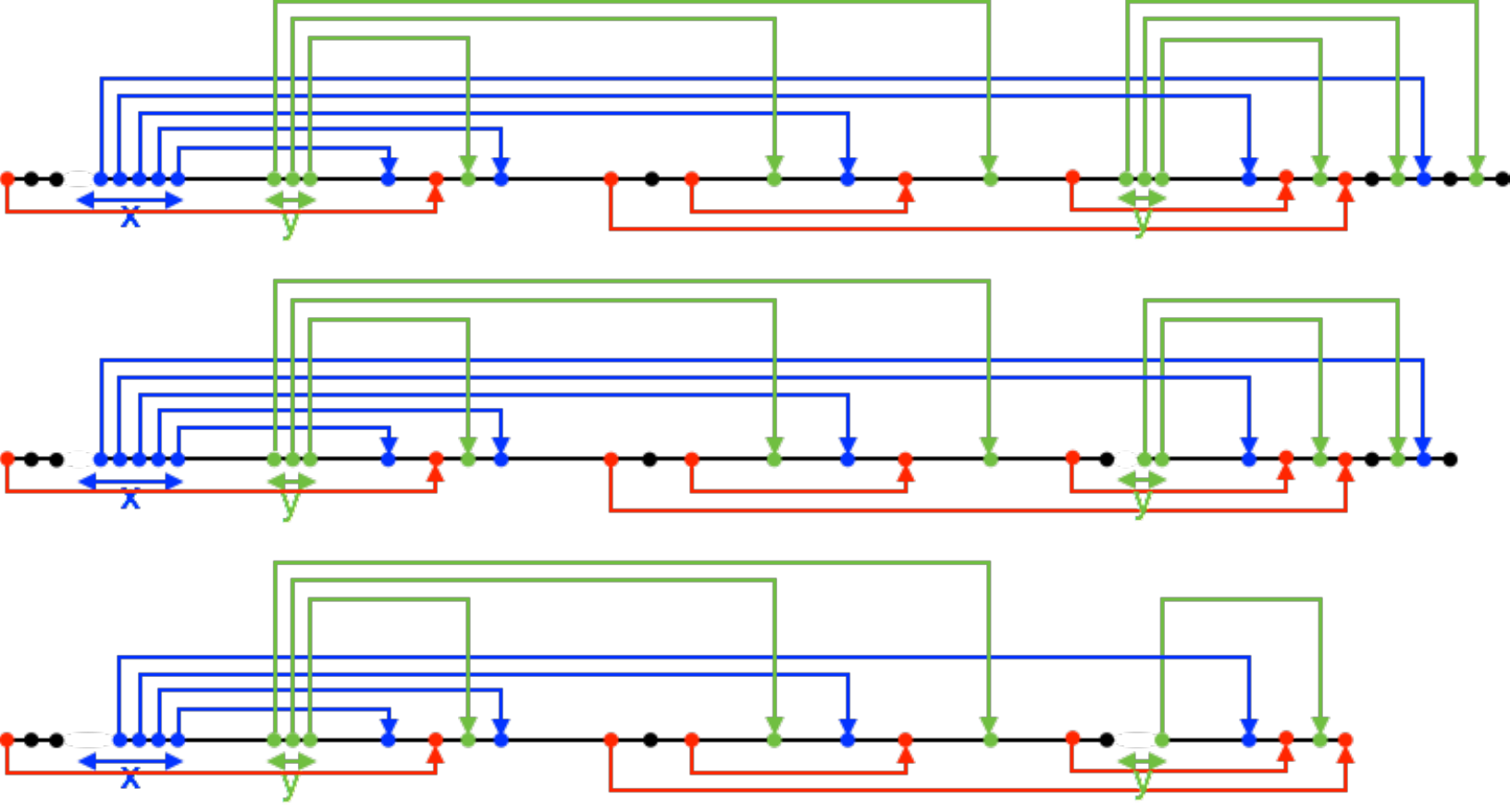}
  \caption{Removing timing constraints. At most one reset hole per clock. Edges below are stack edges. Clock edges are labeled with $x,y$.}
  \label{fig:reset-hole}
\end{figure}

\vspace*{-1ex}
\subparagraph{\bf{Removing a push-pop edge}}
\label{app:push-pop-rem}
Assume now that the last event is a pop: we have $i\matchrel^s j$ where $j$ is
the last point of the split-\TCW.
If there is already a hole immediately before the push event $i$ or if $i$ is the first 
point of the split-\TCW,
then we split after $i$ and before $j$ to detach the atomic matching pair 
$i\matchrel^s j$. This decomposes the here-to-fore connected split-\TCW into two disconnected components (wrt.\ ${\procrel}\cup{\matchrel}$), one containing the atomic matching pair and the remaining split-\TCW, which has a single connected component (again wrt.\ ${\procrel}\cup{\matchrel}$). Adam should choose this remaining (and connected wrt.\ ${\procrel}\cup{\matchrel}$) split-\TCW and the game continues. 

Note that if $i$ is not the first point of the split-\TCW and there is no hole 
before $i$, we cannot
proceed as we did in the case of clock constraints, since this would create a push-hole and the
pushes are not arranged in blocks as the resets.  Hence, removing push-pop edges
as we removed timing constraints would create an unbounded number of holes.
Instead, we split the \TCW just before the matching push event $i$.  Since
push-pop edges are well-nested \eqref{item:T1} and since $j$ is the last point
of the split-\TCW, there are no push-pop edges crossing position $i$:
$i'\matchrel^s j'$ and $i'<i$ implies $j'<i$.  Hence, only clock constraints may
cross position $i$.

Consider some clock $x$ having timing constraints crossing position $i$.  All
these timing constraints come from the last reset block $B_x$ of clock $x$ which
is before position $i$.  Moreover, these resets form the left part of the reset
block $B_x$.  We detach this left part with two splits, one before the reset
block $B_x$ and one after the last reset of block $B_x$ whose timing constraint
crosses position $i$.  We proceed similarly for each clock of $Y$.  Recall that we
also split the \TCW just before the push event $i$. As a result, the \TCW is 
not connected anymore.  Notice that we have used at most $2|Y|+1$ new splits 
to disconnect the \TCW.
For instance, from the bottom split-\TCW of Figure~\ref{fig:reset-hole},
applying the procedure above, we obtain the split-\TCW on top of
Figure~\ref{fig:push-pop-split} which has two connected components. Notice that 
to detach the left part of the reset block of clock $x$ we only used one 
split since there was already a reset hole at the beginning of this block.
The two connected components are depicted separately at the bottom of
Figure~\ref{fig:push-pop-split}.

\begin{figure}[tbp]
  \centering
  \includegraphics[width=8.5cm,page=3]{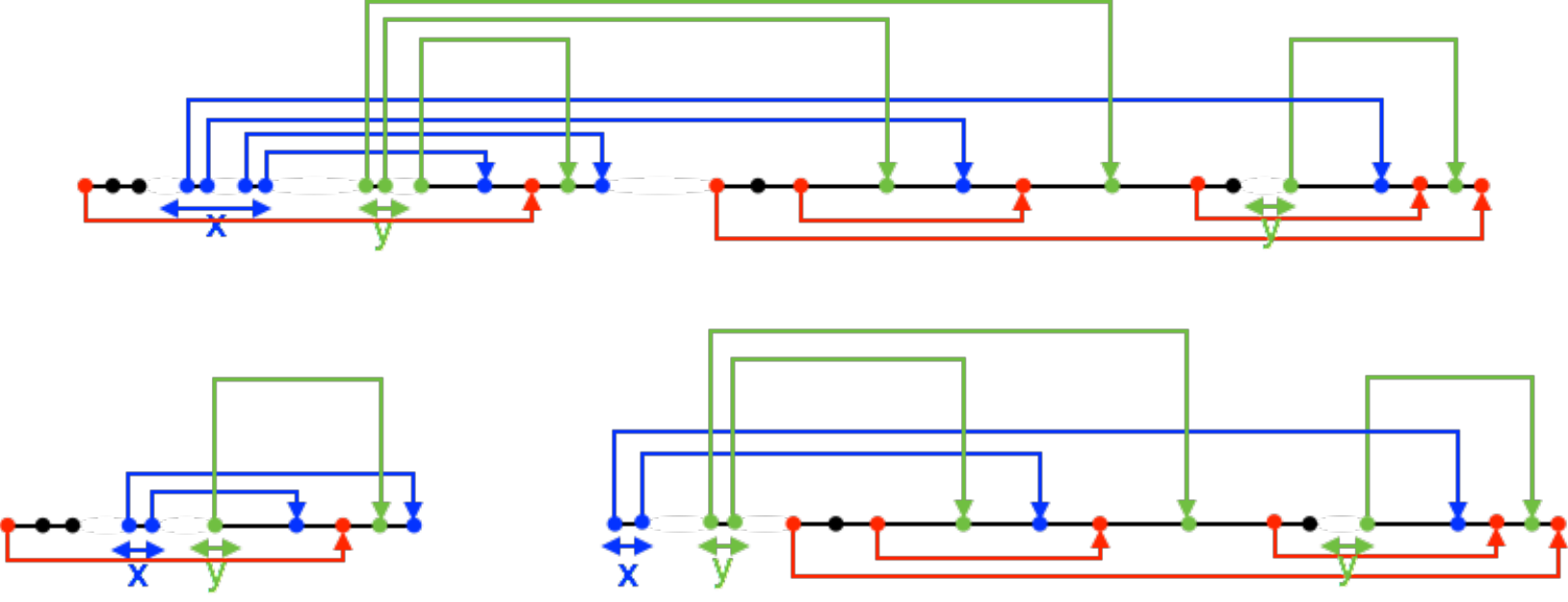}
  \caption{Splitting the \TCW at a position with no crossing stack edge.}
  \label{fig:push-pop-split}
\end{figure}

\subparagraph{\bf{Invariant}.}
\label{app:lem:tpda}
The split-\TCW at the bottom right of Figure~\ref{fig:push-pop-split} is 
representative of the split-\TCW which may occur during the split-game 
using the strategy of \emph{Eve} described above. These split-\TCW 
satisfy the following invariant.

\begin{enumerate}[label=$(\mathsf{I}_{\arabic*})$,ref=$\mathsf{I}_{\arabic*}$]
  \item\label{item:I1}
  The split-\TCW starts with at most one reset block for each clock in $Y$.
  For instance, the split-\TCW of Figure~\ref{fig:invariant-split} starts with two  \emph{reset blocks}, one for clock $x$ and one for clock $y$ (see the two hanging reset blocks on the left, one of $x$ and one of $y$). Indeed, there could be no reset blocks to begin with as well, as depicted in the split-\TCW of Figure~\ref{fig:well-timed}. 
  
  \item\label{item:I2}
  Apart from these reset blocks, the split-\TCW may have reset holes, at most 
  one for each clock in $Y$. 
  For instance, the split-\TCW of Figure~\ref{fig:invariant-split} has two 
  \emph{reset holes}, one for clock $z$ and one for clock $y$.
  
  A reset hole for clock $x$ is followed by the last reset block of clock $x$, if any. Hence, for all timing constraints  $i\matchrel^x j$ such that $j$ is on the right of the hole, the reset event $i$   is in the reset block that starts just after the hole.
\end{enumerate}

\begin{figure}[tbp]
  \centering
  \includegraphics[width=8.5cm,page=4]{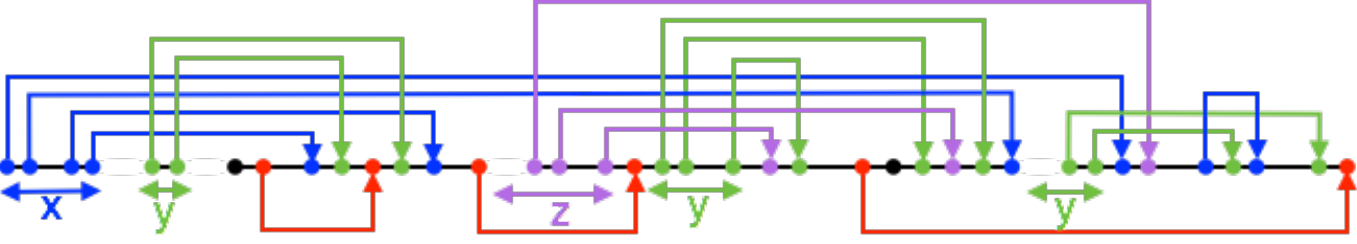}
  \caption{A split-\TCW satisfying (\ref{item:T1}--\ref{item:T2}) and (\ref{item:I1}--\ref{item:I2}).}
  \label{fig:invariant-split}
\end{figure}

\begin{figure}[tbp]
  \centering
  \includegraphics[width=10.5cm,page=5]{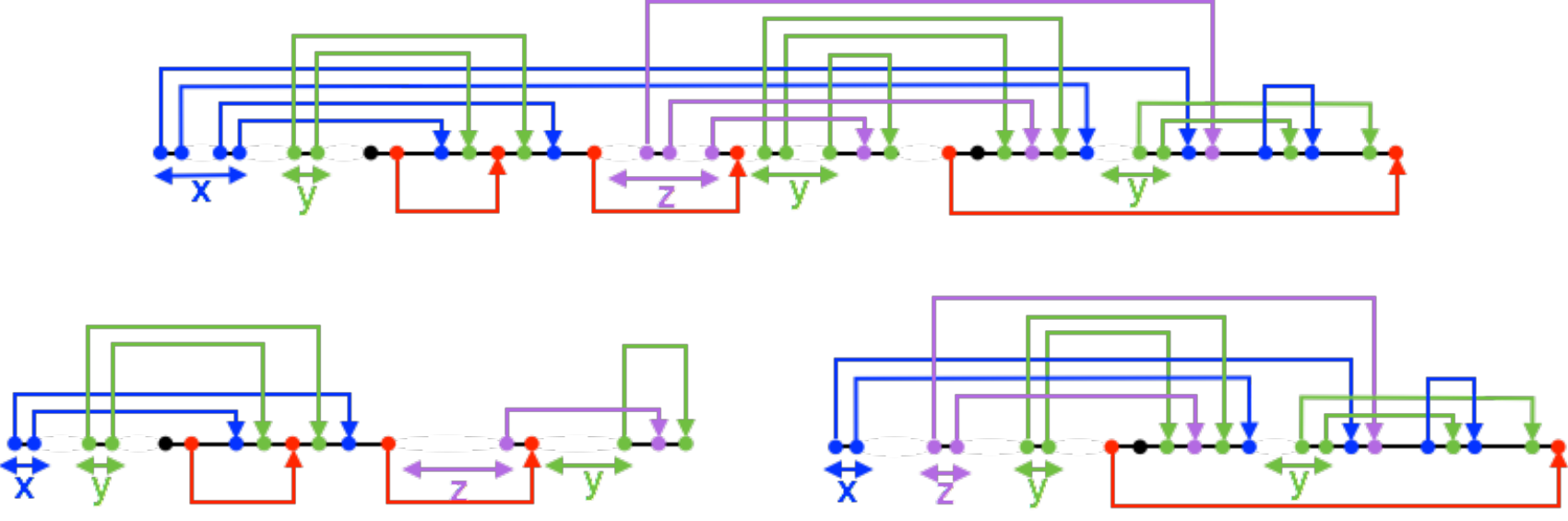}
  \caption{Splitting the \TCW at a position with no crossing stack edge.}
  \label{fig:invariant-split-2}
\end{figure}

\begin{claim}\label{lem:tpda}
  A split-\TCW satisfying (\ref{item:I1}--\ref{item:I2}) has at most $2|Y|+1$
  blocks.  It is disconnected by Eve's strategy using
  at most $2|Y|+1$ new splits, and the resulting connected components satisfy
  (\ref{item:I1}--\ref{item:I2}).  Therefore, its split-width is at most
  $4|Y|+2$.
\end{claim}

\begin{proof}
Let us check that starting from a split-\TCW satisfying
(\ref{item:I1}--\ref{item:I2}), we can apply \emph{Eve's} strategy and the
resulting connected components also satisfy (\ref{item:I1}--\ref{item:I2}).
This is trivial when the last point is internal in which case \emph{Eve}
makes one split before this last point. 

In the second case, the last event checks a timing constraint for some clock
$x$: we have $i\matchrel^x j$ and $j$ is the last event.  If there is already a
reset hole for clock $x$ in the split-\TCW, then by \eqref{item:I2} and
\eqref{item:T2}, the reset event $i$ must be just after the reset hole for clock
$x$. So with at most two new splits \emph{Eve} detaches the atomic edge 
$i\matchrel^x j$ and the resulting split-\TCW satisfies the invariants.
If there is no reset hole for clock $x$ then we consider the last reset block
$B_x$ for clock $x$.  By \eqref{item:T2}, $i$ must be the first event of this
block.  Either $B_x$ is one of the first reset blocks of the split-\TCW
\eqref{item:I1} and \emph{Eve} detaches with at most two splits the atomic
edge $i\matchrel^x j$.  Or \emph{Eve} detaches this atomic edge with at
most three splits, creating a reset hole for clock $x$ in the resulting
split-\TCW. In both cases, the resulting split-\TCW satisfies the invariants.

The third case is when the last event is a pop event: $i\matchrel^s j$ and $j$
is the last event.  Then there are two subcases: either, there is a hole immediately before event $i$ then \emph{Eve} detaches with two splits the atomic edge $i\matchrel^s j$ and the resulting split-\TCW satisfies the invariants.  Or we cut the split-\TCW before position $i$. Notice that no push-pop edges cross $i$: if $i'\matchrel^s j'$ and $i'<i$ then $j'<i$.
As above, for each clock $x$ having timing constraints crossing position $i$, we
consider the last reset block $B_x$ for clock $x$ which is before position $i$.
The resets of the timing constraints for clock $x$ crossing position $i$ form a
left factor of the reset block $B_x$.  We detach this left factor with at most
two splits.  We proceed similarly for each clock of $Y$.  The resulting
split-\TCW is not connected anymore and we have used at most $2|Y|+1$ more
splits.  For instance, if we split the \TCW of Figure~\ref{fig:invariant-split}
just before the last push following the procedure described above, we get the
split-\TCW on top of Figure~\ref{fig:invariant-split-2}.  This split-\TCW is not
connected and its left and right connected components are drawn below.

To see that the invariants are maintained by the connected components, 
let us inspect the splitting of block $B_x$. First, $B_x$ could be one of the beginning reset blocks \eqref{item:I1}. 
This is the case for clock $x$ in Figures~\ref{fig:invariant-split} and~\ref{fig:invariant-split-2}.  
In which case \emph{Eve} use only one
split to divide $B_x$ in $B_x^1$ and $B_x^2$.  The left factor $B_x^1$
corresponds to the edges crossing position $i$ and will form one of the reset
block \eqref{item:I1} of the right connected component.  On the other hand, the
suffix $B_x^2$ stays a reset block of the left connected component.
Second, $B_x$ could follow a reset hole for clock $x$. 
This is the case for clock $z$ in Figures~\ref{fig:invariant-split} and~\ref{fig:invariant-split-2}.  
In which case again 
\emph{Eve} only needs one split to detach the left factor $B_x^1$ which 
becomes a reset block \eqref{item:I1} of the right component. The 
reset hole before $B_x$ stays in the left component.
Finally, assume that $B_x$ is neither a beginning reset block \eqref{item:I1},
nor follows a reset hole for clock $x$.
This is the case for clock $y$ in Figures~\ref{fig:invariant-split} and~\ref{fig:invariant-split-2}.  
Then \emph{Eve} detaches the left 
factor $B_x^1$ which becomes a reset block \eqref{item:I1} of the right component and creates a reset hole in the left component.
\end{proof}

This claim along with the strategy ends the proof of the lemma.
\end{proof}

Now, if the \STCW is from a timed automaton then, $\matchrel^s$ is empty and
Eve's strategy only has the first two cases above.  Thus, we obtain a bound of
$|Y|+3$ on split-width for timed automata.
Thus, we have our second main contribution of this paper and the main theorem of this section, namely, Theorem~\ref{thm:sw-tpda}~\eqref{item:TPDA-sw}.

\begin{theorem}\label{thm:sw-tpda}
  Given a timed system \Sys using a set of clocks $X$, all words in its  \STCW
  language have split-width bounded by $K$, i.e., $\STCW(\Sys)\subseteq\KTCW$,
  where
  \begin{enumerate}
    \item\label{item:TA-sw} $K=|X|+4$ if \Sys is a timed automaton,
    \item\label{item:TPDA-sw} $K=4|X|+6$ if \Sys is a timed pushdown automaton,
  \end{enumerate}
\end{theorem}

\begin{figure}[h!]
  \begin{quote}
  \includegraphics[scale=1,page=8]{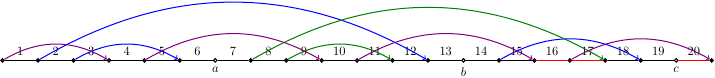}

  \includegraphics[scale=1,page=9]{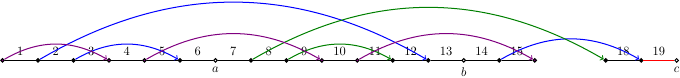}

  \includegraphics[scale=1,page=10]{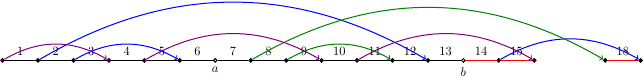}

  \includegraphics[scale=1,page=11]{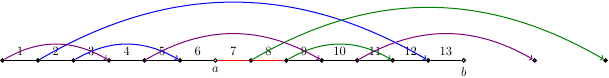}

  \includegraphics[scale=1,page=12]{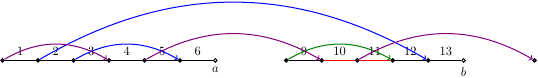}

  \includegraphics[scale=1,page=13]{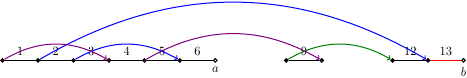}

  \includegraphics[scale=1,page=14]{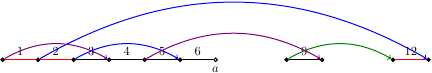}
  \end{quote}

  \caption{The first 4 steps of the split-game on the $\STCW$ of
  Figure~\ref{ex-ta}.  Note that we do not write the time intervals or
  transition names as they are irrelevant for this game. In the first step, 
  Eve detaches the last timing constraint by cutting the 3 successor edges 16, 17 
  and 20. Adam chooses the non atomic \STCW of the second line. In the second 
  step, Eve detaches the last internal event labelled $c$ by cutting edge 19. 
  The resulting \STCW is on the third line, from which Eve detaches the last 
  timing constraint by cutting edges 14, 15, 18. She continues by cutting 
  edges 7, 8 to detach the green timing constraint, and so on. Notice that 
  we have three clocks $Y=\{x,y,\zeta\}$ and that the maximal number of 
  blocks $|Y|+3=6$ is reached on the last \STCW when Eve cuts edges 1, 2 
  and 12.}
  \label{fig:ta-game}
\end{figure} 
Figure~\ref{fig:ta-game} illustrates our strategy starting from the $\STCW$
depicted in Figure~\ref{ex-ta}, generated by the TA shown in that figure.

\newcommand{\BKMV}{\ensuremath{\mathcal{B}^{K,M}_{\textsf{valid}}}\xspace}

\newcommand{\nxt}{\ensuremath{\mathsf{next}}\xspace}
\newcommand{\prv}{\ensuremath{\mathsf{prev}}\xspace}
\newcommand{\sd}{\ensuremath{\mathsf{sd}}\xspace}
\newcommand{\SD}{\ensuremath{\mathsf{SD}}\xspace}
\newcommand{\ts}{\ensuremath{\mathsf{ts}}\xspace}
\newcommand{\tsm}{\ensuremath{\mathsf{tsm}}\xspace}
\newcommand{\bg}{\ensuremath{\mathsf{ac}}\xspace}
\newcommand{\BG}{\ensuremath{\mathsf{AC}}\xspace}
\newcommand{\Bool}{\ensuremath{\mathbb{B}}\xspace}

\section{Tree automata for Validity}
\label{sec:AKMV}

Fix $K,M\geq2$.
Not all graphs defined by $(K,M)$-\STTs are realizable \TCWs.  Indeed, if $\tau$
is such an \STT, the edge relation $\procrel$ may have cycles or may be
branching, which is not possible in a \TCW. Also, the timing constraints given
by $\matchrel$ need not comply with the $\procrel$ relation: for instance, we
may have a timing constraint $e\matchrel f$ with $f\procrel^+ e$.  Moreover,
some terms may define graphs denoting \TCWs which are not realizable.  So we construct the tree 
automaton \AKMV to check for validity.  We will see at the end of the section
(Corollary~\ref{cor:BKMV}) that we can restrict
$\mathcal{A}^{K,M}_{\textsf{valid}}$ further so that it accepts terms denoting
\emph{simple} \TCWs of split-width $\leq K$.

First we show that, since we have only closed intervals in the timing
constraints, considering integer timestamps is sufficient for realizability.

\begin{lemma}\label{lem:integer}
  Let $W=(V,\procrel,(\matchrel^I)_{I\in\mathcal{I}(M)},\lambda)$ be a \TCW using 
  only closed intervals in its timing constraints. Then, $W$ is realizable iff 
  there exists an integer valued timestamp map satisfying all timing 
  constraints.
\end{lemma}

\begin{proof}
  Consider two non-negative real numbers $a,b\in\mathbb{R}_+$ and let
  $i=\lfloor a\rfloor$ and $j=\lfloor b\rfloor$ be their integral parts. Then, 
  $j-i-1<b-a<j-i+1$. It follows that for all closed intervals $I$ with integer 
  bounds, we have $b-a\in I$ implies $j-i\in I$.
  
  Assume there exists a non-negative real-valued timestamp map $\ts\colon
  V\to\mathbb{R}_+$ satisfying all timing constraints of $W$.  From the remark
  above, we deduce that $\lfloor\ts\rfloor\colon V\to\N$ also realizes all
  timing constraints of $W$.
  The converse direction is clear.  
\end{proof}

Consider a set of colors $P\subseteq\{1,\ldots,K\}$.
For each $i\in P$ we let $i^+=\min\{j\in P\cup\{\infty\}\mid i<j\}$
and $i^-=\max\{j\in P\cup\{0\}\mid j<i\}$.
If $P$ is not clear from the context, then we write $\nxt_P(i)$ and $\prv_P(i)$.

We say that a $(K,M)$-\STT is \emph{monotonic} if for every subterm of the form
$\add{i}{j}{\procrel}\tau'$ (resp.\ $\add{i}{j}{\matchrel I}\tau'$ or
$\rename{i}{j}\tau'$) we have $j=\nxt_P(i)$ (resp.\ $i<j$ or
$\prv_P(i)<j<\nxt_P(i)$), where $P$ is the set of active colors in $\tau'$.

In the following, we prove one of our main results viz., the construction 
of the tree automaton \AKMV for checking the validity of the 
\STTs. 
\begin{theorem}\label{thm:AKMV}
We can build a tree automaton \AKMV with $M^{\cO(K)}$ states such that $\Lang(\AKMV)$ is the set of monotonic $(K,M)$-\STTs $\tau$ such that $\sem{\tau}$ is a realizable \TCW and the endpoints of $\sem{\tau}$ are the only colored points.
\end{theorem}

A $(K,M)$-\STT $\tau$ defines a colored graph $\sem{\tau}=(G_\tau,\chi_\tau)$ 
where the graph $G_\tau$ is written $G_\tau=(V,\procrel,(\matchrel^I)_{I\in\mathcal{I}(M)},\lambda)$.
Examples of \STTs and their semantics are given in 
Table~\ref{table:terms-semantics-states}.
The graph $\sem{\tau}$ (or more precisely $G_{\tau}$) is a \TCW if it satisfies several MSO-definable
conditions.  First, $\procrel$ should be the successor relation of a total order
on $V$.  Second, each timing constraint should be compatible with the
total order: ${\matchrel^I}\subseteq{\procrel}^+$.  Also, the \TCW $\sem{\tau}$ 
should be realizable.
These graph properties are MSO definable.  Since a graph
$\sem{\tau}$ has an MSO-interpretation in the tree $\tau$, we deduce that there
is a tree automaton that accepts \emph{all} $(K,M)$-\STTs $\tau$ such that
$\sem{\tau}$ is a realizable \TCW (cf.~\cite{CourcelleBook}).

But this is not exactly what we want/need.  So we provide in the proof below a
direct construction of a tree automaton checking validity.  There are several
reasons for directly constructing the tree automaton $\AKMV$.  First, this
allows to have a clear upper-bound on the number of states of $\AKMV$ (this would be quite technical via MSO).  Second, simplicity of the \TCW and the bound on split-width can be enforced with no additional cost.

\begin{proof} 
  A state of \AKMV will be an abstraction of the graph defined by the \STT read
  so far.  The finite abstraction will keep only the colored points of the
  graph.  We will only accept \emph{monotonic} terms for which the natural order
  on the active colors coincides with the order of the corresponding vertices in
  the final \TCW.
The monotonicity 
ensures that the graph defined by the \STT is in fact a split-\TCW.

\begin{table}
  \noindent%\hspace{-7mm}
  \begin{tabular}{|c|c|c|}
    \hline
    $\tau_1$ & $\tau_2$ & $\tau_3$
    \\ \hline
    \Tree[.$\add{1}{6}{\matchrel[2,\infty]}$ [.$\sttunion$ [.$(1,a)$ ] [.$(6,b)$ ] ] ]
    &
    \Tree[.$\forget{4}$ [.$\add{3}{4}{\procrel}$ [.$\sttunion$
      [.$\add{2}{3}{\matchrel[1,3]}$ [.$\sttunion$ [.$(2,c)$ ] [.$(3,d)$ ] ] ]
      [.$\add{4}{5}{\procrel}$ [.$\add{4}{5}{\matchrel[2,3]}$ [.$\sttunion$ [.$(4,c)$ ] [.$(5,d)$ ] ] ] ]
    ] ] ]
    &
    \Tree[.$\sttunion$
      [.$\rename{3}{5}$ [.$\add{1}{2}{\procrel}$ [.$\forget{5}$ [.$\add{5}{6}{\procrel}$ [.$\sttunion$ [.$\tau_1$ ] [.$\tau_2$ ] ] ] ] ] ]
      [.$\add{3}{4}{\procrel}$ [.$\add{3}{4}{\matchrel[2,\infty]}$ [.$\sttunion$ [.$(3,a)$ ] [.$(4,b)$ ] ] ] ]
    ]
    \rule{9mm}{0pt}
    \\ \hline
    ~\includegraphics[page=5]{gpicture-pics.pdf}%\gusepicture{tau1} 
    &
    \includegraphics[page=7]{gpicture-pics.pdf}%\gusepicture{tau2} 
    &
    ~\includegraphics[page=9]{gpicture-pics.pdf}%\gusepicture{tau3} 
    \\ \hline 
    \rule{0pt}{8mm}%
    \includegraphics[page=6]{gpicture-pics.pdf}%\gusepicture{state1} 
    &
    \includegraphics[page=8]{gpicture-pics.pdf}%\gusepicture{state2} 
    &
    \includegraphics[page=10]{gpicture-pics.pdf}%\gusepicture{state3} 
    \\ \hline  
  \end{tabular}
  
  \caption{The second line gives the tree representations of three monotonic $(6,4)$-\STTs 
  $\tau_1$, $\tau_2$, $\tau_3$, the third 
  line gives their semantics $\sem{\tau}=(G_\tau,\chi_\tau)$ together with a 
  realization $\ts$, the fourth line gives possible states $q$ of \AKMV after 
  reading the terms. The boolean maps $\sd$ and $\bg$ of a state are 
  represented as follows: $\sd(i)=\true$ iff there is a solid edge from $i$ to 
  $i^{+}$, $\bg(i)=\true$ iff the tag ``ac'' is between $\tsm(i)$ and 
  $\tsm(i^{+})$.}
  \protect\label{table:terms-semantics-states}
\end{table}

Moreover, to ensure realizability of the \TCW defined by a term, we will guess 
timestamps of vertices modulo $M$. We also guess 
while reading a subterm whether the time elapsed between two consecutive 
active colors is \emph{big} ($\geq M$) or \emph{small} ($<M$). Then, the 
automaton has to check that all these guesses are coherent and using these 
values it will check that every timing constraint is satisfied.

Formally, \emph{states of \AKMV} are tuples of the form $q=(P,\sd,\tsm,\bg)$,
where $P\subseteq\{1,\ldots,K\}$ are the colors, $\sd,\bg\colon P\to\Bool$ are two boolean-valued functions and $\tsm\colon P\to[M]=\{0,\ldots,M-1\}$ is the time-stamp value modulo $M$. 
The number of states of the tree automaton thus depends on $P$, 
and hence is $M^{\mathcal{O}(K)}$. 

Examples are given in Table~\ref{table:terms-semantics-states}.
The intuition is as follows.  The boolean map $\sd$ is used to check if there is
a solid/process edge (and hence no hole) immediately after $i$, i.e., $\sd(i)=1$
if $i\procrel i^+$.  Further, the boolean tag $\bg(i)$ is set, when the distance
between $i$ and $i^+$ is \emph{accurate}, i.e., can be computed from $\tsm$
values at $i$ and $i^+$ (by subtracting them modulo $M$).  These are maintained
inductively.  More precisely, when reading bottom-up a $(K,M)$-\STT $\tau$ with
$\sem{\tau}=(V,\procrel,\lambda,(\matchrel^I)_{I\in\mathcal{I}(M)},\chi)$, the automaton \AKMV
will reach a state $q=(P,\sd,\ts,\bg)$ such that
\begin{enumerate}[label=($\mathsf{A}_{\arabic*}$),ref=$\mathsf{A}_{\arabic*}$]
  \item\label{prop:A1} $P=\dom(\chi)$ is the set of \emph{active} (non
  forgotten) colors in $\tau$ and 
  
  $\EP(\sem{\tau})\subseteq\chi(P)$: the
  endpoints of $\sem{\tau}$ are colored.
  
  \item\label{prop:A2} For all $i\in P$, we have $\sd(i)=\true$ iff
  $i^+\neq\infty$ and $\chi(i)\procrel^+ \chi(i^+)$ in $\sem{\tau}$.
  
  \item\label{prop:A3} 
  Let ${\hole}=\{(\chi(i),\chi(i^+))\mid i\in P \wedge i^+\neq\infty \wedge
  \sd(i)=\false\}$.  Then, $(\sem{\tau},\hole)$ is a split-\TCW, i.e.,
  ${<}=({\procrel}\cup{\hole})^+$ is a total order on $V$, timing constraints in
  $\sem{\tau}$ are $<$-compatible ${\matchrel}^I\subseteq{<}$ for all $I$, the
  \emph{direct successor} relation of $<$ is ${\lessdot}={\procrel}\cup{\hole}$
  and ${\procrel}\cap{\hole}=\emptyset$.

  \item\label{prop:A4} 
  There exists a timestamp map $\ts\colon V\to\N$ such that
  \begin{itemize}
    \item all constraints are satisfied: $\ts(v)-\ts(u)\in I$ for all $u\matchrel^Iv$

    \item time is non-decreasing: $\ts(u)\leq\ts(v)$ for all $u\leq v$

    \item $(\tsm,\bg)$ is the modulo $M$ abstraction of $\ts$: for all $i\in P$ we
      have
      \begin{itemize}
      \item $\tsm(i)=\ts(\chi(i))[M]$ and,
        \item $\bg(i)=\true$ iff $i^+\neq\infty$, $\ts(\chi(i^+))-\ts(\chi(i))=(\tsm(i^{+})-\tsm(i))[M]$.
          \end{itemize}
  \end{itemize}
\end{enumerate}
We say that $q$ is a \emph{realizable abstraction} of $\tau$ if it satisfies conditions
(\ref{prop:A1}--\ref{prop:A4}).
The intuition behind~(\ref{prop:A4}) is that the finite state automaton \AKMV
cannot store the timestamp map $\ts$ witnessing realizability.  Instead, it
stores the modulo $M$ abstraction $(\tsm,\bg)$ i.e., at each $i$, we store the
modulo $M$ time-stamp of $i$, and a bit $\bg(i)$ that represents whether the
distance between $i$ and its successor can be computed accurately using the
time-stamp values modulo $M$.  We will next see that \AKMV can check
realizability based on the abstraction $(\tsm,\bg)$ of $\ts$ and can maintain
this abstraction while reading the term bottom-up.

We introduce some notations.
Let $q=(P,\sd,\tsm,\bg)$ be a state and let $i,j\in P$ with $i\leq j$.
We define $d(i,j)=(\tsm(j)-\tsm(i))[M]$ and 
$D(i,j)=\sum_{k\in P\mid i\leq k<j}d(k,k^+)$. 
We also define $\BG(i,j)=\bigwedge_{k\in P\mid i\leq k<j}\bg(k)$.
If the state is not clear from the context, then we write $d_q(i,j)$,
$D_q(i,j)$, $\BG_q(i,j)$.
For instance, with the state $q_3$ corresponding to the term $\tau_3$ of
Table~\ref{table:terms-semantics-states}, we have $\BG(1,5)=\true$, $d(1,5)=1$
and $D(1,5)=5=\ts(5)-\ts(1)$ is the \emph{accurate} value of the time elapsed.
Whereas, $\BG(2,6)=\false$ and $d(2,6)=0$, $D(2,6)=4$ are both \emph{strict
modulo-$M$ under-approximations} of the time elapsed $\ts(6)-\ts(2)=8$.

\begin{claim}\label{claim:real-abstr}
  Let $q$ be a state and $\tau$ be an \STT. Then the following hold:
  \begin{enumerate}
    \item Assume (\ref{prop:A1}--\ref{prop:A3}) are satisfied.  Then, for all
    $i,j\in P$ we have $i<j$ iff $\chi(i)<\chi(j)$: the natural ordering on
    colors coincide with the ordering of colored points in the split-\TCW
    $(\sem{\tau},\hole)$.
  
    \item Assume that $\ts$ is a timestamp map satisfying items 2 and 3 of
    \eqref{prop:A4}. Then, for all $i,j\in P$ such that $i\leq j$, we have
    $d(i,j)=D(i,j)[M]=(\ts(\chi(j))-\ts(\chi(i)))[M]$ and
    $d(i,j)\leq D(i,j)\leq\ts(\chi(j))-\ts(\chi(i))$ 
    ($d$ and $D$ give modulo
    $M$ under-approximations of the actual time elapsed). Moreover, 
    $\BG$ tells whether $D$ gives the \emph{accurate} elapse of time:
    \begin{align*}
      \bg(i)=\true & \Longleftrightarrow d(i,i^+)=\ts(\chi(i^+))-\ts(\chi(i)) 
      \\
      \BG(i,j)=\true & \Longleftrightarrow D(i,j)=\ts(\chi(j))-\ts(\chi(i)) \\
      \BG(i,j)=\false & \implies \ts(\chi(j))-\ts(\chi(i)) \geq M 
    \end{align*}
  \end{enumerate}
\end{claim}

\begin{proof}
  1.  From (\ref{prop:A2}--\ref{prop:A3}) we immediately get $\chi(i)<\chi(i^+)$
  for all $i\in P$ such that $i^+\neq\infty$.  By transitivity we obtain
  $\chi(i)<\chi(j)$ for all $i,j\in P$ with $i<j$.  Since $<$ is a strict total
  order on $V$, we deduce that, if $\chi(i)<\chi(j)$ for some $i,j\in P$, then
  $j\leq i$ is not possible.
  
  \smallskip\noindent
  2. Let $i,j\in P$ with $i\leq j$. Using items 2 and 3 of \eqref{prop:A4} we get
  $$
  d(i,j) = (\tsm(j)-\tsm(i))[M] = (\ts(\chi(j))[M]-\ts(\chi(i))[M])[M] = 
  (\ts(\chi(j))-\ts(\chi(i)))[M] \,.
  $$
  Applying this equality for every pair $(k,k^+)$ such that $i\leq k<j$ we get 
  $D(i,j)[M]=(\ts(\chi(j))-\ts(\chi(i)))[M]$. Since $\ts$ is non-decreasing 
  (item 2 of \ref{prop:A4}), it follows that 
  $d(i,j)\leq D(i,j)\leq\ts(\chi(j))-\ts(\chi(i))$.
  
  Now, using again \eqref{prop:A4} we obtain $\bg(k)=\true$ iff
  $d(k,k^+)=\ts(\chi(k^+))-\ts(\chi(k))$. Applying this to all $i\leq k<j$ we 
  get $\BG(i,j)=\true$ iff $D(i,j)=\ts(\chi(j))-\ts(\chi(i))$.
  
  Finally, $\BG(i,j)=\false$ implies $\bg(k)=\false$ for some $i\leq k<j$.  Using
  \eqref{prop:A4} we obtain $\ts(\chi(j))-\ts(\chi(i)) \geq
  \ts(\chi(k^+))-\ts(\chi(k)) \geq M$.
\end{proof}

\begin{table}
  \begin{tabular}{|c|p{120mm}|}
    \hline
    $(i,a)$ 
    &
    $\xrightarrow{(i,a)}q=(P,\sd,\tsm,\bg)$ 
    is a transition if
    $P=\{i\}$, $\sd(i)=\false$ and $\bg(i)=\false$.  
    When reading an atomic \STT $\tau=(i,a)$, only $\tsm(i)$ is guessed.  
    \\ \hline
    $\rename{i}{j}$
    &
    $q=(P,\sd,\tsm,\bg)\xrightarrow{\rename{i}{j}}q'=(P',\sd',\tsm',\bg')$ 
    is a transition if $i\in P$ and $i^-<j<i^+$.  
    Then, $q'$ is obtained from $q$ by replacing $i$ by $j$.  
    \\ \hline
    $\forget{i}$
    &
    $q=(P,\sd,\tsm,\bg)\xrightarrow{\forget{i}}q'=(P',\sd',\tsm',\bg')$ 
    is a transition if
    $i^-,i,i^+\in P$ and $\sd(i^-)=\true=\sd(i)$ (endpoints should stay
    colored).  Then, state $q'$ is deterministically given by
    $P'=P\setminus\{i\}$ and $\bg'(i^-)=\BG(i^-,i^+)\wedge(D(i^-,i^+)<M)$, the
    other values of $\sd',\tsm',\bg'$ being inherited from $\sd,\tsm,\bg$.
    \\ \hline
    $\add{i}{j}{\procrel}$
    &
    $q=(P,\sd,\tsm,\bg)\xrightarrow{\add{i}{j}{\procrel}}q'=(P,\sd',\tsm,\bg)$
    is a transition if
    $i,j\in P$, $j=i^+$ and $\sd(i)=\false$ ($\chi(i)$ does not already
    have a $\procrel$ successor).  The update is given by $\sd'(i)=\true$ since
    we have added a $\procrel$-edge between $\chi(i)$ and $\chi(i^+$), and all
    other values are unchanged.
    \\ \hline
    $\add{i}{j}{\matchrel I}$
    &
    $q=(P,\sd,\tsm,\bg)\xrightarrow{\add{i}{j}{\matchrel I}}q$ 
    is a transition if $i,j\in P$, $i<j$ and either ($\BG(i,j)=\true$ and
    $D(i,j)\in I$) or ($\BG(i,j)=\false$ and $I.up=\infty$).
    \\ \hline
    $\oplus$
    &
    $q_1,q_2\xrightarrow{\oplus}q$
    where $q_1=(P_1,\sd_1,\tsm_1,\bg_1)$,
    $q_2=(P_2,\sd_2,\tsm_2,\bg_2)$ and $q=(P,\sd,\tsm,\bg)$
    is a transition if
    the following hold 
    \begin{itemize}
      \item $P_1\cap P_2=\emptyset$: the $\oplus$ operation on \STTs requires that
      the ``active'' colors of the two arguments are disjoint. 
      \item $P=P_1\cup P_2$, $\sd=\sd_1\cup\sd_2$ and $\tsm=\tsm_1\cup\tsm_2$:
      these updates are deterministic.
      \item the $\procrel$-blocks of the two arguments are shuffled according to
      the ordering of the colors.  Formally, we check that it is not possible to
      insert a point from one argument inside a $\procrel$-block of the other
      argument: 
      \newline
      $\forall i\in P_1\qquad  \sd_1(i)\implies \nxt_P(i)=\nxt_{P_1}(i)$
      \newline
      $\forall i\in P_2\qquad \sd_2(i)\implies \nxt_P(i)=\nxt_{P_2}(i)$.
      \item Finally, $\bg$ satisfies $\bg(\max(P))=\false$ and
      \newline
      $\forall i\in P_1\setminus\{\max(P_1)\}\qquad \bg_1(i)\Longleftrightarrow
      \BG_q(i,\nxt_{P_1}(i)) \wedge D_q(i,\nxt_{P_1}(i)) < M$
      \newline
      $\forall i\in P_2\setminus\{\max(P_2)\}\qquad \bg_2(i)\Longleftrightarrow 
      \BG_q(i,\nxt_{P_2}(i)) \wedge D_q(i,\nxt_{P_2}(i)) < M$.
    \end{itemize}
    \\ \hline
    &
    Notice that these conditions have several consequences. 
    \begin{itemize}
      \item For all $i\in P_1$, if $\nxt_{P}(i)=\nxt_{P_1}(i)$ then $\bg(i)=\bg_1(i)$. 
      \newline
      For all $i\in P_2$, if $\nxt_{P}(i)=\nxt_{P_2}(i)$ then $\bg(i)=\bg_2(i)$.
      \item For all $i\in P_1$, if $\nxt_{P}(i)<\nxt_{P_1}(i)=j<\infty$ (some points
      $\nxt_P(i),\ldots,\prv_P(j)$ of the second argument have been inserted in
      the hole $(i,j)$ of the first argument) 
      then $\bg_1(i)=\true$ implies $\bg(i)=\bg(\prv_P(j))=\true$ and
      $\BG_{q_2}(\nxt_P(i),\prv_P(j))=\true$.
      \newline
      For all $i\in P_2$, if $\nxt_{P}(i)<\nxt_{P_2}(i)=j<\infty$ then 
      $\bg_2(i)=\true$ implies $\bg(i)=\bg(\prv_P(j))=\true$, 
      $\BG_{q_1}(\nxt_P(i),\prv_P(j))=\true$.
    \end{itemize}
    \\ \hline
  \end{tabular}
  
  \caption{Transitions of \AKMV.}
  \protect\label{tab:AKMV}
\end{table}

The \emph{transitions of \AKMV} are defined in Table~\ref{tab:AKMV}. 

\begin{lemma}\label{lem:AKMV-tcw}
  Let $\tau$ be a $(K,M)$-\STT and assume that $\AKMV$ has a run on $\tau$
  reaching state $q$.  Then, $\tau$ is monotonic and $q$ is a realizable
  abstraction of $\tau$.
\end{lemma}

\begin{proof}
  The conditions on the transitions for $\rename{i}{j}$, $\add{i}{j}{\procrel}$ 
  and $\add{i}{j}{\matchrel I}$ directly ensure that the term is monotonic.
  We show that (\ref{prop:A1}--\ref{prop:A4}) are maintained by transitions of \AKMV.
\begin{itemize}
  \item Atomic \STTs: Consider a transition $\xrightarrow{(i,a)}q$ of \AKMV.
  
  It is clear that $q$ is a realizable abstraction of the atomic \STT $\tau=(i,a)$.
  
  \item $\rename{i}{j}$:
  Consider a transition $q\xrightarrow{\rename{i}{j}}q'$ of \AKMV.
  
  Assume that $q$ is a realizable abstraction of some $(K,M)$-\STT $\tau$ and 
  let $\tau'=\rename{i}{j} \tau$. It is easy to check that $q'$ is a 
  realizable abstraction of $\tau'$.

  \item $\forget{i}$:
  Consider a transition $q\xrightarrow{\forget{i}}q'$ of \AKMV.

  Assume that $q$ is a realizable abstraction of some $(K,M)$-\STT $\tau$ and 
  let $\tau'=\forget{i} \tau$. It is easy to check that $q'$ is a 
  realizable abstraction of $\tau'$. In particular, the correctness of the 
  update $\bg'(i^-)$ follows from Claim~\ref{claim:real-abstr}.

  \item $\add{i}{j}{\procrel}$:
  Consider a transition $q\xrightarrow{\add{i}{j}{\procrel}}q'$ of \AKMV. 
  
  Assume that $q$ is a realizable abstraction of some $(K,M)$-\STT $\tau$ and 
  let $\tau'=\add{i}{j}{\procrel} \tau$. It is easy to check that $q'$ is a 
  realizable abstraction of $\tau'$.
  
  \item $\add{i}{j}{\matchrel I}$: 
  Consider a transition $q\xrightarrow{\add{i}{j}{\matchrel I}}q$ of \AKMV.

  Assume that $q$ is a realizable abstraction of some $(K,M)$-\STT $\tau$ with 
  the timestamp map $\ts$. It is easy to check that $q'$ is a 
  realizable abstraction of $\tau'=\add{i}{j}{\matchrel I} \tau$ with the same 
  timestamp map $\ts'=\ts$. We only have to check the properties for the new 
  timing constraint. 
  First, condition $i<j$ and Claim~\ref{claim:real-abstr} ensures that
  the new timing constraint is compatible with the linear order ${<}'={<}$.

  Second, by Claim~\ref{claim:real-abstr}, if $\BG(i,j)=\false$ then the time 
  elapsed between $\chi(i)$ and $\chi(j)$ is big and
  the timing constraint $I$ is satisfied iff $I.up=\infty$.  On the other hand, 
  if $\BG(i,j)=\true$ then using again Claim~\ref{claim:real-abstr} we deduce 
  that $D(i,j)$ is the actual time elapsed between $\chi(i)$ and $\chi(j)$ and 
  the timing constraint $I$ is satisfied iff $D(i,j)\in I$.
  
  \item $\oplus$:
  Consider a transition $q_1,q_2\xrightarrow{\oplus}q$ of \AKMV.

    Assume that $q_1$ and $q_2$ are realizable abstractions of some 
    $(K,M)$-\STTs $\tau_1$ and $\tau_2$ with timestamp maps $\ts_1$ and $\ts_2$ 
    respectively. Let $\tau=\tau_1\oplus\tau_2$. We show that $q$ is a 
    realizable abstraction of $\tau$. Notice that $\sem{\tau}$ is the disjoint 
    union of $\sem{\tau_1}$ and $\sem{\tau_2}$.
    
    \medskip\noindent\eqref{prop:A1} %
    We have $\dom(\chi)=\dom(\chi_1)\cup\dom(\chi_2)=P_1\cup P_2=P$.
    
    Moreover, 
    $\EP(\sem{\tau})=\EP(\sem{\tau_1})\cup\EP(\sem{\tau_2})
    \subseteq\chi(P_1)\cup\chi(P_2)=\chi(P)$.
    
    \medskip\noindent\eqref{prop:A2} %
    Let $i\in P_1$ be such that $j=\nxt_P(i)=\nxt_{P_1}(i)$.  We have
    $\sd(i)=\true$ iff $\sd_1(i)=\true$ iff $j\neq\infty$ and
    $\chi_1(i)\procrel^+ \chi_1(j)$ in $\sem{\tau_1}$ iff $j\neq\infty$ and
    $\chi(i)\procrel^+ \chi(j)$ in $\sem{\tau}$.
    
    Let $i\in P_1$ be such that $j=\nxt_P(i)<\nxt_{P_1}(i)$.  We have
    $\sd(i)=\sd_1(i)=\false$, $j\neq\infty$ and there are no $\procrel$-path
    from $\chi(i)=\chi_1(i)$ to $\chi(j)=\chi_2(j)$ in $\sem{\tau}$ since
    $\sem{\tau_1}$ and $\sem{\tau_2}$ are disjoint.
    
    We argue similarly for $i\in P_2$ which concludes the proof of 
    Condition~\eqref{prop:A2}.
    
    \medskip\noindent\eqref{prop:A3} %
    Let ${\hole}=\{(\chi(i),\chi(j))\mid i\in P \wedge j=\nxt_P(i)\neq\infty
    \wedge \sd(i)=\false\}$.  Let ${<}=({\procrel}\cup{\hole})^+$.  Using
    \eqref{prop:A2} and the definition of $<$, it is easy to see that for all
    $i,j\in P$, if $i<j$ then $\chi(i)<\chi(j)$.  We deduce that
    ${<_1}\cup{<_2}\subseteq{<}$.
    
    Let $u\matchrel^I v$ be a timing constraint in $\sem{\tau}$.  Either it is
    in $\sem{\tau_1}$ and it is compatible with $<_1$, hence also with $<$.  Or
    it is in $\sem{\tau_2}$ and it is compatible with $<_2$ and with $<$.
    
    Next, we show that blocks of $\sem{\tau_1}$ and $\sem{\tau_2}$ do not 
    overlap in $(\sem{\tau},\hole)$. Let $i_1,j_1$ be the colors of the left and right 
    endpoints of some block $\chi(i_1)\procrel^*\chi(j_1)$ in $\sem{\tau_1}$.
    For all $k\in P_1$ such that $i_1\leq k<j_1$ we have
    $\chi(i_1)\procrel^*\chi(k)\procrel^+\chi(\nxt_{P_1}(k))\procrel^*\chi(j_1)$.
    Applying \eqref{prop:A2} we get $\sd_1(k)=\true$.  Now, using the definition
    of the transition for $\oplus$, we get $\nxt_P(k)=\nxt_{P_1}(k)$.  We deduce
    by induction that for all $\ell\in P_2$, if $i_1<\ell$ then $j_1<\ell$.
    By symmetry, the same holds for a block of $\sem{\tau_2}$ and a color of
    $P_1$. 
    
    Let $i_2,j_2$ be the colors of the left and right endpoints of some
    block $\chi(i_2)\procrel^*\chi(j_2)$ in $\sem{\tau_2}$.
    Either $i_1<i_2$ and we get $j_1<i_2$ hence $\chi(j_1)<\chi(i_2)$.  
    Or $i_2<i_1$ and we get $j_2<i_1$ hence $\chi(j_2)<\chi(i_1)$.  
    We deduce that the blocks of $\sem{\tau_1}$ and $\sem{\tau_2}$ are shuffled
    in $(\sem{\tau},\hole)$ according to the order of the colors of their endpoints.
    
    Therefore, $<$ is a total order on $V$ and $(\sem{\tau},\hole)$ is a 
    split-\TCW.
        
    \medskip\noindent\eqref{prop:A4} %
    We construct the timestamp map $\ts$ for $\tau$ inductively on 
    $V=V_1\uplus V_2$ following the successor relation 
    ${\lessdot}={\procrel}\cup{\hole}$. If 
    $v=\min(V)$ is the first point of the split-\TCW, we let 
    $$
    \ts(v)=\begin{cases}
      \ts_1(v) & \text{if } v\in V_1 \\
      \ts_2(v) & \text{if } v\in V_2 \,.
    \end{cases}
    $$
    Next, if $\ts(u)$ is defined and $u\procrel v$ then we let
    $$
    \ts(v)=\begin{cases}
      \ts(u) + \ts_1(v)-\ts_1(u) & \text{if } u\in V_1 \\
      \ts(u) + \ts_2(v)-\ts_2(u) & \text{if } u\in V_2 \,.
    \end{cases}
    $$
    Finally, if $\ts(u)$ is defined and $u\hole v$ then, with $i,j\in P$ 
    being the colors of $u$ and $v$ ($\chi(i)=u$ and $\chi(j)=v$), we let
    $$
    \ts(v)=\begin{cases}
      \ts(u) + d_q(i,j) & \text{if } \bg(i)=\true \\
      \ts(u) + d_q(i,j)+M & \text{if } \bg(i)=\false \,.
    \end{cases}
    $$
    With this definition, the following hold
    \begin{itemize}
      \item Time is clearly non-decreasing: $\ts(u)\leq\ts(v)$ for all $u\leq v$
      
      \item $(\tsm,\bg)$ is the modulo $M$ abstraction of $\ts$.  The proof is
      by induction. 
      
      First, if $i=\min(P)$ then $v=\min(V)=\chi(i)$ and $i\in
      P_1$ iff $v\in V_1$.  Using the definitions of $\tsm$ and $\ts$, we
      deduce easily that $\tsm(i)=\ts(v)[M]$.  
      
      Next, let $i\in P$ with $j=\nxt_P(i)<\infty$.  Let $u=\chi(i)$,
      $v=\chi(j)$ and assume that $\tsm(i)=\ts(u)[M]$.
      
      If $u\procrel^+v$ and $u\in V_1$ then $v\in V_1$, $i,j\in P_1$, 
      $\sd_1(i)=\true$ and 
      \begin{align*}
        \ts(v)[M] &= (\ts(u)[M] + \ts_1(v)[M]-\ts_1(u)[M])[M] \\
        &= (\tsm(i)+\tsm_1(j)-\tsm_1(i))[M]
        = \tsm(j) \,.
      \end{align*}
      Moreover, $\bg(i)=\bg_1(i)$ and $\ts(v)-\ts(u)=\ts_1(v)-\ts_1(u)$. We 
      deduce that $\bg(i)=\true$ iff $\ts(v)-\ts(u)< M$.
      The proof is similar if $u\procrel^+v$ and $u\in V_2$.
      
      Now, if $u\not\procrel^+v$ then $u\hole v$ (endpoints are always 
      colored). We deduce that
      $$
      \ts(v)[M] = (\ts(u)[M] + d_q(i,j))[M] = (\tsm(i)+\tsm(j)-\tsm(i))[M] =
      \tsm(j) \,.
      $$
      Moreover, it is clear that $\ts(v)-\ts(u)< M$ iff $\bg(i)=\true$.
            
      \item Constraints are satisfied. Let $u\matchrel^Iv$ be a timing 
      constraint in $\sem{\tau}=\sem{\tau_1}\uplus\sem{\tau_2}$. Wlog we 
      assume that $u,v\in V_1$. We know that $\ts_1(v)-\ts_1(u)\in I$.
      
      If $u\procrel^+v$ then we get $\ts(v)-\ts(u)=\ts_1(v)-\ts_1(u)$ from the
      definition of $\ts$ above.  Hence, $\ts(v)-\ts(u)\in I$.
      
      Now assume there are holes between $u$ and $v$ in $(\sem{\tau},\hole)$.
      Let $u'$ be the right endpoint of the block of $u$ and $v'$ be the left 
      endpoint of the block of $v$. Since endpoints are colored we find 
      $i,j\in P$ such that $u'=\chi(i)$ and $v'=\chi(j)$. We have $i,j\in 
      P_1$, $i<j$, $u\procrel^* u'$ and $v'\procrel^* v$. We deduce from the 
      definition of $\ts$ that $\ts(u')-\ts(u)=\ts_1(u')-\ts_1(u)$ and
      $\ts(v)-\ts(v')=\ts_1(v)-\ts_1(v')$. Now, using Claim~\ref{claim:oplus} 
      below we obtain:
      \begin{itemize}
        \item Either $\BG_{q_1}(i,j)=\false$ and $\ts(v')-\ts(u')\geq M$. From 
        Claim~\ref{claim:real-abstr} we also have $\ts_1(v')-\ts_1(u')\geq 
        M$. We deduce that $I.up=\infty$ and $\ts(v)-\ts(u)\in I$.
        
        \item Or $\BG_{q_1}(i,j)=\true$ and 
        $\ts(v')-\ts(u')=\ts_1(v')-\ts_1(u')$. Therefore, 
        \begin{align*}
          \ts(v)-\ts(u) & = \ts(v)-\ts(v') + \ts(v')-\ts(u') + \ts(u')-\ts(u) \\
          & = \ts_1(v)-\ts_1(v') + \ts_1(v')-\ts_1(u') + \ts_1(u')-\ts_1(u) \\
          & = \ts_1(v)-\ts_1(u)\in I \tag*{\qEd}
        \end{align*}
      \end{itemize}
    \end{itemize}
  \end{itemize}
  \def\popQED{}
\end{proof}
      
\begin{claim}\label{claim:oplus}
  Let $i,j\in P_1$ with $i\leq j$ and let $u=\chi(i)$ and $v=\chi(j)$.
  \begin{enumerate}
    \item If $\BG_{q_1}(i,j)=\false$ then $\ts(v)-\ts(u)\geq M$.
    
    \item If $\BG_{q_1}(i,j)=\true$ then $\ts(v)-\ts(u)=\ts_1(v)-\ts_1(u)$.
  \end{enumerate}
\end{claim}

\begin{proof}
  The proof is by induction on the number of points in $P_1$ between $i$
  and $j$.  The result is clear if $i=j$.  So assume that
  $k=\nxt_{P_1}(i)\leq j$ and let $w=\chi(k)$.  By induction, the claim
  holds for the pair $(k,j)$.
  \begin{enumerate}
    \item If $\BG_{q_1}(i,j)=\false$ then either $\bg_1(i)=\false$ or 
    $\BG_{q_1}(k,j)=\false$. 
    
    In the first case, by definition of the transition for $\oplus$, we
    have either $\BG_q(i,k)=\false$ or $D_q(i,k)\geq M$.  In both cases,
    we get $\ts(w)-\ts(u)\geq M$ by Claim~\ref{claim:real-abstr}.
    
    In the second case,  we get $\ts(v)-\ts(w)\geq M$ by induction.
    
    Since $\ts$ is non-decreasing, we obtain $\ts(v)-\ts(u)\geq M$.
    
    \item If $\BG_{q_1}(i,j)=\true$ then $\bg_1(i)=\true$ and 
    $\BG_{q_1}(k,j)=\true$. 
    
    By induction, we obtain $\ts(v)-\ts(w)=\ts_1(v)-\ts_1(w)$.
    
    From the definition of the transition for $\oplus$, since 
    $\bg_1(i)=\true$, we get $\BG_q(i,k)=\true$ and $D_q(i,k)<M$. 
    Using Claim~\ref{claim:real-abstr} we deduce that 
    $\ts(w)-\ts(u)=D_q(i,k)$. Now, $D_q(i,k)<M$ implies 
    $D_q(i,k)=d_q(i,k)=d_{q_1}(i,k)$. Using again 
    Claim~\ref{claim:real-abstr} we get 
    $d_{q_1}(i,k)=\ts_1(w)-\ts_1(u)$. We conclude that             
    $\ts(w)-\ts(u)=\ts_1(w)-\ts_1(u)$.
    
    Combining the two equalities, we obtain 
    $\ts(v)-\ts(u)=\ts_1(v)-\ts_1(u)$ as desired.
    \qedhere
  \end{enumerate}
\end{proof}

\paragraph{Accepting condition} %
The accepting states of \AKMV should correspond to abstractions of \TCWs.  Hence
the accepting states are of the form $(\{i\},\sd,\tsm,\bg)$ (reached while
reading an atomic term $(i,a)$) or $(\{i,j\},\sd,\tsm,\bg)$ with
$i,j\in\{1,\ldots,K\}$, $i<j$, $\sd(i)=\true$ and $\sd(j)=\false=\bg(j)$.

We show the correctness of the construction.

\medskip\noindent$(\subseteq)$ 
Let $\tau$ be an \STT accepted by \AKMV. There is an accepting run of \AKMV
reading $\tau$ and reaching state $q$ at the root of $\tau$.  By 
Lemma~\ref{lem:AKMV-tcw}, the term $\tau$ is monotonic and the state $q$ is a
realizable abstraction of $\tau$, hence $(\sem{\tau},\hole)$ is a split-\TCW.
But since $q$ is accepting, we have ${\hole}=\emptyset$.  Hence $\sem{\tau}$ is
a \TCW. Moreover, from \eqref{prop:A4} we deduce that $\sem{\tau}$ is 
realizable and the endpoints of $\sem{\tau}$ are the only colored points by 
\eqref{prop:A1} and the acceptance condition.

\medskip\noindent$(\supseteq)$ 
Let $\tau$ be a monotonic $(K,M)$-\STTs such that $\sem{\tau}=(G,\chi)$ is a realizable
\TCW and the endpoints of $\sem{\tau}$ are the only colored points.  
Let $\ts\colon V\to\N$ be a timestamp map satisfying all the timing constraints
in $\tau$.  We construct a run of $\AKMV$ on $\tau$ by resolving the
non-deterministic choices as explained below.  Notice that the transitions for
$\rename{i}{j}$, $\forget{i}$, $\add{i}{j}{\procrel}$ and 
$\add{i}{j}{\matchrel I}$ are deterministic.
We will obtain an accepting run $\rho$ of \AKMV on $\tau$ such that for every
subterm $\tau'$, the state $\rho(\tau')$ satisfies \eqref{prop:A4} with
timestamp map $\ts$, or more precisely, with the restriction of $\ts$ to the
vertices in $\sem{\tau'}$.

\begin{itemize}
  \item A leaf $(i,a)$ of the term $\tau$ corresponds to some vertex $v\in V$.
  The transition taken at this leaf will set $\tsm(i)=\ts(v)[M]$ so that
  \eqref{prop:A4} holds with $\ts$.

  \item We can check that the conditions enabling transitions at
  $\rename{i}{j}$, $\forget{i}$ or $\add{i}{j}{\procrel}$ nodes are satisfied
  since $\tau$ is monotonic and $\sem{\tau}$ is a \TCW whose endpoints are
  colored.

  \item For a subterm $\tau'=\add{i}{j}{\matchrel I}\tau''$ the condition $i<j$
  is satisfied since $\tau$ is monotonic and the condition envolving $\BG$ and
  $D$ is satisfied by Claim~\ref{claim:real-abstr} since $\ts$ satisfies all
  timing constraints of $\tau$ and \eqref{prop:A4} holds with $\ts$ at $\tau''$.

  \item Consider a subterm $\tau'=\tau_1\oplus\tau_2$.  Let
  $\rho(\tau_1)=q_1=(P_1,\sd_1,\tsm_1,\bg_1)$ and
  $\rho(\tau_2)=q_2=(P_2,\sd_2,\tsm_2,\bg_2)$.  The active colors of $\tau_1$
  and $\tau_2$ are disjoint, hence $P_1\cap P_2=\emptyset$ by \eqref{prop:A1}.
  Define $q'=(P',\sd',\tsm',\bg')$ by $P'=P_1\cup P_2$, $\sd'=\sd_1\cup\sd_2$,
  $\tsm'=\tsm_1\cup\tsm_2$ and for all $i\in P'$, $\bg'(i)=\true$ iff 
  $i^{+}\neq\infty$ and $\ts(\chi'(i^{+}))-\ts(\chi'(i))<M$.
  
  The condition for $q_1,q_2\xrightarrow{\oplus}q'$ on shuffling the
  $\procrel$-blocks holds since $\sem{\tau}$ is a \TCW. Now, we look at the
  condition on $\bg'$.  Let $i\in P_1\setminus\{\max(P_1)\}$ and
  $j=\nxt_{P_1}(i)$.  We have $\bg_1(i)=\true$ iff
  $\ts(\chi_1(j))-\ts(\chi_1(i))<M$ since \eqref{prop:A4} holds with $\ts$
  at $\tau_1$.  The latter holds iff for all $k\in P'$ with $i\leq k<j$ we have
  $\ts(\chi'(k^{+}))-\ts(\chi'(k))< M$ (i.e., $\BG_{q'}(i,j)=\true$ by the
  above definition of $\bg'$) and $D_{q'}(i,j)< M$ (again, by the definition
  of $\bg'$ we have that $\bg'(k)=\true$ implies
  $d_{q'}(k,k^{+})=\ts(\chi'(k^{+}))-\ts(\chi'(k))$ and $\BG'(i,j)=\true$ implies
  $D_{q'}(i,j)=\ts(\chi'(j))-\ts(\chi'(i))$).
  \qedhere
\end{itemize}
\end{proof}

To ensure that we accept only terms denoting simple \TCWs of split-width at most $K$, we adapt our above construction 
to $\mathcal{A}^{2K,M}_{\textsf{valid}}$. We will refer this restriction 
to $\mathcal{A}^{2K,M}_{\textsf{valid}}$ as \BKMV. 

\begin{corollary}\label{cor:BKMV}
  We can build a tree automaton \BKMV with $M^{\cO(K)}$ states such that $\tau\in \Lang(\BKMV)$ iff $\sem \tau \in \KMTCW$ is realizable.
\end{corollary}

\begin{proof}
  We will construct $\BKMV$ as a restriction of $\mathcal{A}=\mathcal{A}^{2K,M}_{\textsf{valid}}$.
  
  Let $\tau$ be a term accepted by $\mathcal{A}$.  Notice that, for every
  subterm $\tau'$ of $\tau$, the number of blocks in $\sem{\tau'}$ can be
  computed from the state $q'=(P',\sd',\tsm',\bg')$ that labels $\tau'$ in this
  accepting run: this number of blocks is $|\{i\in P'\mid\sd(i)=\false\}|$.  So
  we restrict $\mathcal{A}$ to the set of states $q'$ such that the number of
  blocks of $q'$ is at most $K$.  The automaton \BKMV is defined by further
  restricting $\mathcal{A}$ so that adding a timing constraint
  $\add{i}{j}{\matchrel I}$ is allowed only on subterms of the form
  $(i,a)\oplus(j,b)$ with $i<j$.
  
  Then, an accepted term $\tau\in\Lang(\BKMV)$ is realizable since it is accepted by $\cA$ (since $\BKMV$ is a restriction of $\cA$) and it describes a split-decomposition of $\sem{\tau}=(\stcw,\chi)$ of width at most $K$.  Further, since timing   constraints are added only on terms of the form $(i,a)\oplus(j,b)$, the \TCW   $\sem \tau$ must be simple.  We obtain $\sem{\tau}\in\KMTCW$ and is realizable.
  
  Conversely, let $\stcw$ be a realizable simple \TCW in $\KMTCW$.  There is a split-tree
  $T$ of width at most $K$ for $\stcw$.  By Lemma~\ref{lem:split-tcw-stt}, there is a
  $(2K,M)$-\STT $\tau$ with $\sem{\tau}=(\stcw,\chi)$. We may
  assume that $\tau$ is monotonic and only the endpoints of $\stcw$ are colored
  by $\chi$. By Theorem~\ref{thm:AKMV}, we get $\tau\in\Lang(\mathcal{A})$.
  Moreover, for every subterm $\tau'$ of $\tau$, the graph $\sem{\tau'}$ has at
  most $K$ blocks.  Also, the timing constraints in $\tau$ are only added by 
  subterms of the form $\add{i}{j}{\matchrel I}((i,a)\oplus(j,b))$ with $i<j$.
  Therefore, $\tau\in\Lang(\BKMV)$.
\end{proof}

%%% Local Variables:
%%% mode: latex
%%% TeX-master: "main.tex"
%%% End:

\newcommand{\src}{\ensuremath{\mathsf{src}}\xspace}
\newcommand{\tgt}{\ensuremath{\mathsf{tgt}}\xspace}
\newcommand{\clr}{\ensuremath{\mathsf{clr}}\xspace}
\newcommand{\rg}{\ensuremath{\mathsf{rgc}}\xspace}
\newcommand{\pup}{\ensuremath{\mathsf{pp}}\xspace}

%$\st, \clr, \rg, \pup$

\section{Tree automata for timed systems}
\label{sec:TA-sys}

The goal of this section is to build a tree automaton which accepts the \STTs denoting (realizable) \STCWs accepted by a TPDA. Again, to prove the existence of a tree automaton, it suffices to show the MSO-definability of the existence of a (realizable) run of \Sys on a simple \TCW, and appeal to Courcelle's theorem~\cite{CourcelleBook}.  However, as explained in the previous section, this may not give an optimal complexity.  Hence, we choose to directly construct the tree automaton \AKMS, where $M$ is one more than the maximal constant used in 
the TPDA, and $K$ is the split width. Our input is therefore the system, in this case a TPDA $\Sys$, whose size $|\Sys|$ includes the number of states, transitions and the maximal constant used in $\Sys$.

Let us explain first how the automaton would work for an \emph{untimed} system with no stack. At a leaf of the \STT of the form $(i,a)$, the tree automaton guesses a transition $\delta$ from \Sys which could be executed reading action $a$.  It keeps the transition in its state, paired with color $i$.  After reading a
subterm $\tau$, the tree automaton stores in its state, for each active color $k$ of $\tau$, the pair of source and target states of the transition $\delta_k$ guessed at the corresponding leaf, denoted $\src(k)$ and $\tgt(k)$, respectively.  This information can be easily updated at nodes labelled $\oplus$ or $\forget{i}$ or $\rename{i}{j}$.  When, reading a node labelled $\add{i}{j}\procrel$, the tree automaton checks that  $\tgt(i)$ equals $\src(j)$.  This ensures that the transitions guessed at leaves form a run when taken in the total order induced by the word denoted by the final term.  At the root, we check that at most two colors remain active: $i$ and $j$ for the leftmost (resp.\ rightmost) endpoint of the word.  Then, the tree automaton accepts if $\src(i)$ is an initial state of \Sys and $\tgt(j)$ is a final state of \Sys.

%a map $\Delta$ which assigns with each active color $k$ of $\tau$ the transition $\Delta(k)$ guessed at the corresponding leaf.  The map $\Delta$ can be easily updated at nodes labelled $\oplus$ or $\forget{i}$ or $\rename{i}{j}$.  When, reading a node labelled $\add{i}{j}\procrel$, the tree automaton checks that the target state of the transition $\Delta(i)$ equals the source state of the transition $\Delta(j)$.  This ensures that the transitions guessed at leaves form a run when taken in the total order induced by the word denoted by the final term.  At the root, we check that at most two colors remain free: $i$ and $j$ for the leftmost (resp.\ rightmost) endpoint of the word.  Then, the tree automaton accepts if the source state of $\Delta(i)$ is initial in \Sys and the target state of $\Delta(j)$ is final in \Sys.

The situation is a bit more complicated for timed systems.  First, we are
interested in the \emph{simple} \TCW semantics in which each event is blown-up
in several micro-events.  Following this idea, each transition $\delta=(s,\gamma,a,\op,R,s')$ of the timed system \Sys is blown-up in micro-transitions as explained in Section~\ref{sec:tpda-sem}, assuming that $\gamma=\gamma_1\wedge\cdots\wedge\gamma_n$ has $n$ conjuncts of the form $x\in I$, and that $R=\{x_1,x_2,\ldots,x_m\}$:
\begin{center}
  \includegraphics[width=85mm,page=1]{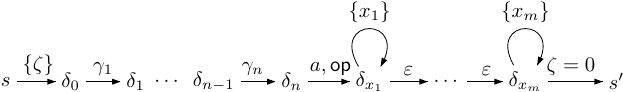}
\end{center}
Notice that the new states of the micro-transitions, e.g., $\delta_i$ or
$\delta_x$, uniquely identify the transition $\delta$ and the guard $\gamma_i$
or the clock $x$ to be reset. Notice also that the reset-loop at the end allows an arbitrary number (possibly zero) resets of clock $x_1$ (depending on
how many timing constraints for clock $x_1$ will originate from this reset),
followed by an arbitrary number (possibly zero) resets of clock $x_2$, etc. Further, all transitions except, the \emph{middle} one are $\epsilon$-transitions; and in what follows, we will call this transition (labelled $a,\op$) as the \emph{middle micro-transition} of this sequence or of $\delta$.

The difficulty now is to make sure that, when a guard of the form $x\in I$ is
checked in some transition, then the source point of the timing constraint in
the simple \TCW indeed corresponds to the latest transition resetting clock $x$.
To check this property, the tree automaton \AKMS stores,
\begin{itemize}
  \item a map $\clr$ from colors in $P$ to sets of clocks.  For each color $k$
  of the left endpoint of a block, $\clr(k)$ is the set of clocks reset by the
  transition whose middle micro-transition (the one corresponding to $a,\op$
  micro-transition) occurs in the block.

  \item a map $\rg$ from pairs of colors in $P$ to sets of clocks.  For each
  pair $(i,j)$, $i<j$, of colors of left endpoints of \emph{distinct} blocks,
  the set $\rg(i,j)$ contains the set of clocks that are reset in block $i$ and
  checked in $j$.
\end{itemize}

Finally, the last property to be checked is that the push-pop edges are
well-nested.  To this end, the tree automaton \AKMS stores the set $\pup$ of
pairs of colors $(i,j)$, $i<j$, of left endpoints of \emph{distinct} blocks such
that there is at least one push-pop edge from block $i$ to block $j$.

Formally, a state of \AKMS is a tuple $(q,\src,\tgt,\clr,\rg,\pup)$ where $q=(P,\ldots)$ is a state of \BKMV (as seen in Corollary~\ref{cor:BKMV})
%Theorem \ref{thm:AKMV}) 
the maps $\src,\tgt$ assign to each color $k\in P$, respectively, the source and target states of the micro-transition guessed at the leaf corresponding to
color $k$, and the maps $\clr$, $\rg$ and the set $\pup$ are as described above.  
The number of states of \AKMS is calculated as follows:
\begin{itemize}
\item $q$ depends on $P$; the number of possible subsets $P$ is $2^{\mathcal{O}(K)}$.
\item The size of $\src,\tgt$ is $|\Sys|^K$, % where $S$ denotes number of states of the underlying system $\Sys$.   
\item The size of $\clr$ is $[2^{|X|}]^K$, while the size of $\rg$ is $[2^{|X|}]^{K^2}$, and the size of $\pup$ is $2^{K^2}$. 
\end{itemize}
The number of states of \AKMS is hence bounded by 
$|\Sys|^{\mathcal{O}(K)}\times 2^{\mathcal{O}(K^2(|X|+1))}$. 

\begin{table}[h!]
  \begin{tabular}{|c|p{120mm}|}
    \hline
    $(i,a)$ & $\xrightarrow{(i,a)} (q,\src,\tgt,\clr,\rg,\pup)$ is a transition of \AKMS
    if $\xrightarrow{(i,a)} q$ is a transition of $\BKMV$ and
%     \begin{itemize}
%     \item     either 
    there is a transition $\delta=(s,\gamma,a,\nop,R,s')$ of \Sys (which is
    guessed by the tree automaton \AKMS) such that:
    \begin{itemize}
    \item %$\Delta(i)=(\delta_n,\delta_{x})$ is
      the \emph{middle} micro-transition of $\delta$ (assuming $\gamma$ has $n$ conjuncts)
      is from $\delta_n=\src(i)$ to $\delta_{x_1}=\tgt(i)$, 
      \item $\clr(i)=R$ since the set of clocks reset by $\delta$ is $R$, and 
      \item $\rg$ is nowhere defined and $\pup=\emptyset$ since we have a single block.
    \end{itemize}
%     \item 
    Notice that if $a=\varepsilon$, \AKMS may guess the special initial dummy
    transition $\delta=(s_{dummy},\true, \varepsilon,\nop, X,s_0)$ where $s_0$
    is an initial state of $\Sys$, to simulate the reset of all clocks when the 
    run starts.
%     \end{itemize}

    \\ \hline
    $\rename{i}{j}$
    &
    $(q,\src,\tgt,\clr,\rg,\pup)\xrightarrow{\rename{i}{j}}(q',\src',\tgt',\clr',\rg',\pup')$ is a
    transition of \AKMS if $q\xrightarrow{\rename{i}{j}}q'$ is a transition of
    \BKMV and $\src',\tgt'$, $\clr'$, $\rg'$, $\pup'$ are obtained from $\src,\tgt$, $\clr$, $\rg$, $\pup$, resp., by replacing $i$ by $j$.

    \\ \hline
    $\forget{i}$
    & 
    $(q,\src,\tgt,\clr,\rg,\pup)\xrightarrow{\forget{i}}(q',\src',\tgt',\clr',\rg',\pup')$ is a
    transition of \AKMS if $q\xrightarrow{\forget{i}}q'$ is a transition of
    \BKMV (in particular, $i$ is not an endpoint) and $\src',\tgt'$ are obtained from
    $\src,\tgt$ by forgetting the entry of color $i$.  Since $i$ is not a left
    endpoint, $\clr,\rg$ and $\pup$ are not affected: $\clr'=\clr$, $\rg'=\rg$ and $\pup'=\pup$.
    \\ \hline
    {\small $\add{i}{j}{\matchrel I}((i,\varepsilon)\oplus (j,\varepsilon))$}
    & 
    $\xrightarrow{\add{i}{j}{\matchrel I}((i,\varepsilon)\oplus
    (j,\varepsilon))} (q,\src,\tgt,\clr,\rg,\pup)$ is a transition of \AKMS if
    $\xrightarrow{\add{i}{j}{\matchrel I}((i,\varepsilon)\oplus
    (j,\varepsilon))} q$ is a transition of \BKMV and either
    \begin{itemize}
      \item there exists a clock constraint $x\in I$ induced by two transitions $\delta^1=(s^1,\gamma^1,a^1,\op^1,R^1,s'^1)$ and $\delta^2=(s^2,\gamma^2,a^2,\op^2,R^2,s'^2)$ of \Sys, i.e., $x\in R^1$ and some conjunct of $\gamma^2$, say the $k$-th, is $x\in I$ (the tree automaton \AKMS guesses this) such that:
      \begin{itemize}
      \item %$\(i)=(\delta^1_{x},\delta^1_{x})$ is
        the reset micro-transition for clock $x$ of $\delta^1$ satisfies $\delta^1_{x}=\src(i)=\tgt(i)$,
\item %        $\Delta(j)=(\delta^2_{k-1},\delta^2_k)$
        the micro-transition checking the $k$-th conjunct of $\gamma^2$ is from $\src(j)=\delta^2_{k-1}$ to $\tgt(j)=\delta^2_k$
        \item $\clr(i)=\clr(j)=\emptyset$, $\rg(i,j)=\{x\}$ and $\pup=\emptyset$.
      \end{itemize}
      \item or $I=[0,0]$ and the tree automaton \AKMS guesses that it encodes
      the $\zeta$ clock constraint of some transition
      $\delta=(s,\gamma,a,\op,R,s')$ of \Sys.  Then, $\src(i)=s,
      \tgt(i)=\delta_0$, $\src(j)=\delta_{x},\tgt(j)=s')$, where $\delta_0$ and
      $\delta_x$ are the first and last states of the micro-transitions from $\delta$,
      $\clr(i)=\clr(j)=\emptyset$, $\rg(i,j)=\emptyset$ and $\pup=\emptyset$.
     \end{itemize}
    \\ \hline
\end{tabular}
  \caption{Transitions of \AKMS -- part 1}
  \protect\label{tab:AKMSa}
\end{table}

\begin{table}
  \begin{tabular}{|c|p{120mm}|}
    \hline
{\small  $\add{i}{j}{\matchrel I}((i,a^1)\oplus(j,a^2))$}
  & 
  $\xrightarrow{\add{i}{j}{\matchrel I}((i,a^1)\oplus(j,a^2))} (q,\src,\tgt,\clr,\rg,\pup)$
  is a transition of \AKMS if $\xrightarrow{\add{i}{j}{\matchrel
  I}((i,a^1)\oplus(j,a^2))}q $ is a transition of \BKMV and there are two
  matching push-pop transitions of \Sys:
  $\delta^1=(s^1,\gamma^1,a^1,\push_b,R^1,s'^1)$ and
  $\delta^2=(s^2,\gamma^2,a^2,\pop_b^{I},R^2,s'^2)$ (these transitions are
  guessed by \AKMS) such that %$\Delta(i)=(\delta^1_{n},\delta^1_{x})$ is
  the \emph{middle} micro-transition of $\delta^1$ is from $\delta^1_n=\src(i)$
  to $\delta^1_{x_1}=\tgt(i)$,
  the \emph{middle} micro-transition of $\delta^2$ is from $\delta^2_m=\src(j)$
  to $\delta^2_{y_1}=\tgt(j)$,
%  $\Delta(j)=(\delta^2_{m},\delta^2_{y})$ is
%  the \emph{middle} micro-transition of $\delta^2$,
  $\clr(i)=R^1$, $\clr(j)=R^2$, $\rg(i,j)=\emptyset$ and $\pup=\{(i,j)\}$.
  \\ \hline

    $\add{i}{j}{\procrel}$
    &
    $(q,\src,\tgt,\clr,\rg,\pup)\xrightarrow{\add{i}{j}{\procrel}}(q',\src',\tgt',\clr',\rg',\pup')$ is
    a transition of \AKMS if $q\xrightarrow{\add{i}{j}{\procrel}}q'$ is a transition of \BKMV (thus, $i$ is a right endpoint and $j=i^{+}$ is a left 
    endpoint) and
    \begin{itemize}
      \item Either $\tgt(i)=\src(j)$, or there is an $\varepsilon$-path of
      micro-transitions
      $\tgt(i)\xrightarrow{\varepsilon}\delta_{x_k}\xrightarrow{\varepsilon}\cdots
      \xrightarrow{\varepsilon}\delta_{x_\ell}\xrightarrow{\varepsilon}\src(j)$.
      Indeed, adding $\procrel$ between $i$ and $j$ means that these points are
      consecutive in the final \TCW, hence it should be possible to concatenate
      the micro-transition taken at $i$ and $j$.

      \item $\src'=\src, \tgt'=\tgt$.

      \item Let $k$ be the left end-point of $\mathsf{Block}(i)$. 
    We merge $\mathsf{Block}(k)$ and $\mathsf{Block}(j)$
      hence we set $\clr'(k)=\clr(k)\cup \clr(j)$ and $\clr'(\ell)=\clr(\ell)$ if
      $\ell\notin\{k,j\}$ is another left endpoint.  Also, if
      $\ell,\ell'\notin\{k,j\}$ are other left endpoints, we set
      $\rg'(\ell,\ell')=\rg(\ell,\ell')$ if $\ell<\ell'$, $\rg'(\ell,k)=\rg(\ell,k)\cup
      \rg(\ell,j)$ if $\ell<k$, and $\rg'(k,\ell)=\rg(k,\ell)\cup \rg(j,\ell)$ if
      $j<\ell$.  Finally, $\pup'$ is the set of pairs $(\ell,\ell')$ of left
      endpoints of distinct blocks in $q'$ s.t.\ either $(\ell,\ell')\in \pup$, or
      $\ell=k$ and $(j,\ell')\in \pup$, or $\ell<k=\ell'$ and $(\ell,j)\in \pup$.
    \end{itemize}
    \\ \hline
    $\oplus$
    &
    $(q_1,\src_1,\tgt_1,\clr_1,\rg_1,\pup_1),(q_2,\src_2,\tgt_2,\clr_2,\rg_2,\pup_2)\xrightarrow{\oplus}(q,\src,\tgt,\clr,\rg,\pup)$
    is a transition of \AKMS if $q_1,q_2\xrightarrow{\oplus}q$ is a transition
    of \BKMV and
    \begin{itemize}
      \item The other components  are inherited:
      $\src=\src_1\cup\src_2$, $\tgt=\tgt_1\cup \tgt_2$, $\clr=\clr_1\cup \clr_2$, $\rg=\rg_1\cup \rg_2$ and
      $\pup=\pup_1\cup \pup_2$.

      \item For all $i<k<j$ left endpoints in $q$, we have $\clr(k)\cap
      \rg(i,j)=\emptyset$.  This is the crucial condition which ensures that a
      timing constraint always refers to the last transition resetting the clock
      being checked.  If some clock $x\in \rg(i,j)$ is reset in $\mathsf{Block}(i)$ and
      checked in $\mathsf{Block}(j)$, it is not possible to insert  $\mathsf{Block}(k)$ between
      blocks of $i$ and $j$ if clock $x$ is reset by a transition in $\mathsf{Block}(k)$.

      \item For all $(i,j)\in \pup$ and $(k,\ell)\in \pup$, if $i<k<j$ then $\ell\leq j$.
      This ensures that the push-pop edges are well-nested.
    \end{itemize}
    \\ \hline
  \end{tabular}
  \caption{Transitions of \AKMS -- part 2}
  \protect\label{tab:AKMSb}
\end{table}

The transitions of \AKMS will check the existence of transitions of $\BKMV$ to check for validity and in addition will check the existence of an abstract run of the system $\Sys$. These are described in detail in Tables~\ref{tab:AKMSa}, \ref{tab:AKMSb}. With this, we finally have the following acceptance condition. A state $(q,\src,\tgt,\clr,\rg,\pup)$ is accepting if 
\begin{itemize}
  \item $q$ is an accepting state of \BKMV, hence it consists of a single block 
  with left endpoint $i$ and right endpoint $j$ (possibly $i=j$),
  \item $\src(i)$ is an initial state of \Sys, and $\tgt(j)$ is a final state of \Sys. %, or is a micro-transition state of the form $\delta_x$ where the target state of the corresponding transition $\delta$ is a final state of \Sys.
\end{itemize}

If the underlying system is a timed automaton, it is sufficient to store the maps $\clr,\rg$ in the state, since there are no stack operations.  From the above construction we obtain the following theorem:

\begin{restatable}{theorem}{propAKMS}\label{prop:AKMS}
  Let \Sys be a TPDA of size $|\Sys|$ 
  with set of clocks $X$ and using constants less than $M$.  Then, we
  can build a tree automaton \AKMS of size $|\Sys|^{\mathcal{O}(K)}.2^{\mathcal{O}(K^2(|X|+1))}$ such that
  $\Lang(\AKMS) = \{\tau\in\Lang(\BKMV)\mid \sem{\tau}\in \STCW(\Sys)\}$.
\end{restatable}
\begin{proof}
  Let $\stt$ be a \kSTT and $\stcw=\sem{\stt}$. We will show that $\stt$ is accepted by \AKMS iff $\stcw\in \Lang(\BKMV)$ and $\stcw\in\STCW(\Sys)$.

  Assume that \AKMS has an accepting run on $\stt$.  Each move of \AKMS first checks if there is a move on \BKMV, hence we obtain an accepting run of $\BKMV$ on $\stt$. Thus, $\stcw\in\Lang(\BKMV)$ and so by Corollary~\ref{cor:BKMV} of Theorem~\ref{thm:AKMV}, $\stcw$ is a realizable \STCW. It remains to check that $\stcw$ is generated or accepted by $\Sys$. 
  That is, we need to show that there exists an abstract accepting run of \Sys on $\stcw$ which starts with an initial state, ends in a final state and satisfies the three conditions stated in Section~\ref{sec:tpda-sem}.

  \begin{enumerate}
    \item We first define the sequence of transitions.  Each vertex $v$ of \stcw labeled by letter $a$ is introduced as a node colored $i$ in some atomic term $(i,a)$.     We let $\delta(v)$ be the transition guessed by \AKMS when reading this atomic term. 
      We start by reading $\varepsilon$ and guessing the special initial transition that resets all clocks and goes to an initial state $s_0$. Subsequently,  for each $u\procrel v$ in \stcw, we will have an $\add{i}{j}{\procrel}$ occurring in $\stt$ such that $\mathsf{source}(\delta(u))=\src(i)$, $\mathsf{target}(\delta(v))=\tgt(j)$ and either $\tgt(i)=\src(j)$ or there is a (possibly empty) sequence of $\varepsilon$ micro-transitions from $\tgt(i)$ to $\src(j)$. Note that the existence of such a sequence of micro-transitions can indeed be recovered from the structure of this sequence, since the information about the transition is fully contained in the micro-transition sequence (including the set of clocks resets and checked etc). As a result, we have constructed a sequence of transitions $(\delta(v))_v$ which forms a path in \Sys reading \stcw and starting from the initial state $s_0$.  By the acceptance condition, if $v$ is the maximal vertex of \stcw then  $\mathsf{target}(\delta(v))$ is a final state.

\item Now, we need to check that the sequence of push-pop operations is well-nested. This is achieved by using $\pup$ which, as mentioned earlier, stores the pairs $(i,j)$ of left endpoints of distinct blocks such that there is at least one push-pop edge from block $i$ to block $j$. Note that every such push-pop edge gets into the set $\pup$ by the last atomic rule in Table~\ref{tab:AKMSa}, where it is checked that the transitions that were guessed at $i$ and $j$ were indeed, matching push-pop transitions. This information is propagated correctly when we add a process edge by the last condition in $\add{i}{j}{\procrel}$ rule. That is, if the adding of an edge changes the left end-point, then the elements of $\pup$ are updated to refer to the new left end-point. This ensures that the information is not forgotten since (left) end-points are never forgotten. Now, observe that the only operation that may allow the push-pop edges to cross is the general combine. But the last condition, in the definition of transition here, uses the set $\pup$ information to make sure that well-nesting is not violated. More precisely, if $(i,j), (k,\ell)$ is in $\pup_1\cup \pup_2$ and $i<k<j$, we check that $\ell\leq j$ and hence the edges are well-nested.

\item Finally, we check that the clock constraints are properly matched (i.e., when a guard is checked, the source of the timing constraint is indeed the latest transition resetting this clock. To do this, we make use of the maps $\clr$ and $\rg$. The map $\clr$ maintains for each left-endpoint of a block, the set of clocks reset by transitions whose middle micro-transitions occur in that block. This is introduced at the atomic transition rule and maintained during all transitions and propagated correctly when left-endpoints change, which happens only when blocks are merged by $\add{i}{j}{\procrel}$ rule (see last condition in this rule).  The set $\rg(i,j)$ refers to the set of clocks that are reset in $\mathsf{Block}(i)$ and checked at $\mathsf{Block}(j)$. Again, this set is populated at an atomic matching rule  of the form  $\add{i}{j}{\matchrel I}((i,\varepsilon)\oplus (j,\varepsilon))$ as defined. And as before, it is maintained throughout and propagated when end-points of blocks change, i.e., during the adding of a process edge.

Now, for any guard $x\in I$ that is checked, say matching vertex $u$ to $v$, we
claim that (i) $x$ is reset at $u$ and checked at $v$ wrt $x\in I$ and (ii)
between $u$ and $v$, $x$ is never reset.  (i) follows from the definition of the
atomic rule for adding a matching edge relation.  And for (ii), we observe that
once the matching relation is added between $u,v$, then $x\in \rg(i,j)$ where
$i,j$ are respectively the colors of the left endpoints of the blocks containing
$u$ and $v$.  If the general combine is not used, then we cannot add a reset in
between, and hence (ii) holds.  When the general combine is used between state
$q_1$ containing $i,j$ and another state $q_2$, then we check whether
$\clr(k)\cap\rg(i,j)=\emptyset$ for each color $k$ of left endpoint of a block
in $q_2$ with $i<k<j$.  This ensures that any block $k$ inserted between $i$ and
$j$ ($x\in \rg(i,j)$ means that there is a matching edge for clock $x$ between
blocks of $i, j$) cannot contain a transition which resets the clock $x$ ($x$ is
stored in $\clr(k)$).  Thus, the last reset-point is preserved and our checks
are indeed correct.
   \end{enumerate}
  Thus, we obtain that $\stcw$ is indeed generated by $\Sys$, i.e., $\stcw\in   \STCW(\Sys)$.

  In the reverse direction, if $\stcw\in \STCW(\Sys)\cap \Lang(\BKMV)$, then there is a sequence of transitions which lead to the accepting state on reading $\stcw$ in $\Sys$. By guessing each of these transitions correctly at every point (at the level of the atomic transitions), we can easily generate the run of our automaton \AKMS which is also consequently accepting. 
\end{proof}
As a corollary we obtain,
\begin{corollary}
\label{thm:AKMS}
Given a timed (pushdown) automaton $\Sys$, $\Lang(\Sys)\neq \emptyset$ iff $\Lang(\AKMS) \neq \emptyset$.
\end{corollary}
\begin{proof}
If $\Lang(\Sys) \neq \emptyset$, then there is a realizable \STCW $\stcw$ accepted by $\Sys$. 
By Theorem~\ref{thm:sw-tpda}, we know that its split-width is bounded by a constant $K=4|X|+6$. Further, from the proof of Theorem \ref{thm:sw-tpda}, we obtain that Eve has a winning strategy on \stcw with atmost $K-1$ holes. Eve's strategy is such that  it generates a $K$-\STT $\tau$  such that  $\stcw=\sem{\stt}$. Since $\stcw=\sem \stt$ is realizable, by Theorem \ref{thm:AKMV}, $\tau \in \Lang(\BKMV)$. Thus, $\stcw$ in $\STCW(\Sys)$ and $\tau \in \Lang(\BKMV)$. Theorem~\ref{prop:AKMS} above then gives  $\stcw \in \Lang(\AKMS)$.  The converse argument is similar.
\end{proof}

As a consequence of all the above results, we obtain the decidability and 
complexity of emptiness for timed (pushdown) automata. These results were first 
proved in \cite{AD94,lics12}.

\begin{theorem}\label{thm:final-sys}
  Checking emptiness for TPDA $\Sys$ is decidable in \textsc{ExpTime}, while
  timed automata are decidable in \textsc{PSpace}.
\end{theorem}

\begin{proof}
  We first discuss the case of TPDA. First note that all \STCWs in the semantics
  of a TPDA have a split-width bounded by some constant $K$.  By Corollary
  \ref{thm:AKMS}, checking non-emptiness of $\Sys$ boils down to checking
  non-emptiness of tree automata \AKMS. But while checking non-emptiness of
  \AKMS, the construction also appeals to \BKMV and checks that its
  non-emptiness too.  From the fact that checking emptiness of tree automata is
  in \textsc{PTIME} and given the sizes of the constructed tree automata \BKMV,
  \AKMS (both exponential in the size of the input TPDA $\Sys$), we obtain an
  \textsc{ExpTime} procedure to check emptiness of $\Sys$.

  If the underlying system $\Sys$ was a timed automaton, then as seen in Figure
  \ref{fig:ta-game}, Eve's strategy gives rise to a word-like decomposition of
  the \STCW. Hence, the split-trees are actually word-like binary trees, where
  one child is always an atomic node.  Thus, we can use word automata (to guess
  the atomic node) instead of tree automata to check emptiness.  Using the
  \textsc{NLOGSPACE} complexity of emptiness for word automata, and using the
  exponential sizes of \BKMV, \AKMS, we obtain a \textsc{PSPACE} procedure to
  check emptiness of $\Sys$.
\end{proof}

%%% Local Variables:
%%% mode: latex
%%% TeX-master: "main.tex"
%%% End:

\section{Dense time multi-stack pushdown systems}
\label{sec:mpda}
 As another application of our technique, we now consider the model of dense-timed multi-stack pushdown automata (\MPDA), which have several stacks.  The reachability problem for untimed multi-stack pushdown automata (MPDA) is already undecidable, but several restrictions have been studied on (untimed) MPDA, like bounded rounds~\cite{latin10}, bounded phase~\cite{TMP07}, bounded scope~\cite{TNP14}, and so on to regain decidability.  

In this section, we consider $\MPDA$ with the restriction of ``bounded rounds''. To the best of our knowledge, this timed model has not been investigated until now. Our goal is to illustrate how our technique can easily be applied here with a minimal overhead (in difficulty and complexity).  

Formally, a \MPDA is a tuple $\Sys=(S, \Sigma, \Gamma, X, s_0, F, \Delta)$
similar to a TPDA defined in Section~\ref{sec:tpda-sem}. The only difference is the stack operation $\op$ which now specifies which stack is being operated on. That is,
\begin{enumerate}
  \item $\nop$ does not change the contents of any stack (same as before),
  \item $\push^i_c$ where $c\in\Gamma$ is a push operation that adds $c$ on top of stack $i$, with age 0.
  \item $\pop^i_{c \in I}$ where $c\in\Gamma$ 
  and $I \in \mathcal{I}$   is a pop operation that removes the top most symbol of stack $i$ provided it
  is a $c$ with age in the interval $I$.
\end{enumerate}
A sequence $\sigma=\op_1\cdots\op_m$ of operations is a \emph{round} if it can
be decomposed in $\sigma=\sigma_1\cdots\sigma_n$ where each factor $\sigma_i$ is
a possibly empty sequence of operations of the form $\nop$, $\push^i_c$, 
$\pop^i_{c\in I}$.

Let us fix an integer bound $k$ on the number of rounds.  The semantics of the
\MPDA in terms of \STCWs is exactly the same as for TPDA, except that the
sequence of stack operations along any run is restricted to (at most) $k$
rounds.  Thus, any run of \MPDA can be broken into a finite number of contexts,
such that in each context only a single stack is used.  As before, the sequence of
push-pop operations of any stack must be well-nested.

\subsection*{Simple TC-word semantics for \MPDA}
\label{stc-tpda}

We define the semantics for  \MPDA in terms of simple \TCWs. 
Let $n$ denote the number of stacks.  
A simple \TCW $\stcw=(V,\procrel,(\matchrel^I)_{I\in\mathcal{I}},\lambda)$
is said to be generated or accepted by a \MPDA $\Sys$ if there is an accepting
abstract run $\rho=(s_0,\gamma_1,a_1,\op_1,R_1,s_1)$
$(s_1,\gamma_2,a_2,\op_2,R_2,s_2)\cdots$ $(s_{m-1},\gamma_m,a_m,\op_m,R_m,s_m)$
of $\Sys$ such that $s_m\in F$ is a final state.

Let $\stcw=(V,\procrel,(\matchrel^I)_{I\in\mathcal{I}},\lambda)$
 be a simple \TCW. Recall that
${<}={\procrel}^+$ is the transitive closure of the successor relation.  We say
that \stcw is $k$-round \emph{well timed} with respect to a set of clocks $X$
and stacks $1\leq s\leq n$ if for each $I\in \cI$, the $\matchrel^I$ relation for timing constraints
can be partitioned as ${\matchrel^I}=\biguplus_{1\leq s\leq
n}{\matchrel^{s\in I}}\uplus\biguplus_{x\in X}{\matchrel^{x\in I}}$ where
\begin{enumerate}[label=$(\mathsf{T}'_{\arabic*})$,ref=$\mathsf{T}'_{\arabic*}$]
  \item\label{item:Tp1}
  for each $1\leq s\leq n$,
  the relation $\matchrel^s=\bigcup_{I\in \cI}\matchrel^{s\in I}$ corresponds to the matching push-pop events of stack $s$, 
  hence it is well-nested: for all $i\matchrel^s j$ and $i'\matchrel^s j'$, if 
  $i<i'<j$ then $i'<j'<j$, see Figure~\ref{fig:well-timed}.
  
  \noindent
  Moreover, $\stcw$ consists of at most $k$ rounds, i.e., we have
  $V=V_1\uplus\cdots\uplus V_k$ with $V_i\times V_j\subseteq{<}$ for all $1\leq
  i<j\leq k$.  And each $V_\ell$ is a round, i.e.,
  $V_\ell=V_\ell^1\uplus\cdots\uplus V_\ell^n$ with $V_\ell^s\times
  V_\ell^t\subseteq{<}$ for $1\leq s<t\leq n$ and push pop events of $V_\ell^s$
  are all on stack $s$ (for all $i\matchrel^t j$ with $t\neq s$ we have 
  $i,j\notin V_\ell^s$).
  
  \item\label{item:Tp2}
  An $x$-reset block is a maximal consecutive sequence $i_1\lessdot\cdots\lessdot
  i_n$ of positions in the domain of the relation $\matchrel^x=\bigcup_{I\in \cI}\matchrel^{x\in I}$. 
  For each $x\in X$, the relation $\matchrel^x$ corresponds to the timing
  constraints for clock $x$ and  is well-nested: for all $i\matchrel^x j$ and $i'\matchrel^x j'$, if
  $i<i'$ are in the same $x$-reset block, then $i<i'<j'<j$. 
  Each guard should be matched with the closest 
  reset block on its left: for all $i\matchrel^x j$ and $i'\matchrel^x j'$, if  
  $i<i'$ are not in the same $x$-reset block then $j<i'$.
\end{enumerate}

It is easy to check that the simple \TCWs defined by a $k$-\MPDA are well-timed, 
i.e., satisfy the properties above.

We denote by $\STCW(\Sys)$ the set of simple \TCWs generated by $\Sys$.  
The language of $\mathcal{L}(\Sys)$ is the set of \emph{realizable} simple \TCWs in 
$\STCW(\Sys)$. Given a bound $k$ on the number of rounds, we denote by 
$\STCW(\Sys,k)$ the set of simple \TCWs generated by runs of \Sys using at most 
$k$ rounds. We let $\mathcal{L}(\Sys,k)$ be the corresponding language.
Given a \MPDA $\Sys$, we show that all simple \TCWs in $\STCW(\Sys,k)$ have
bounded split-width.  Actually, we will prove a slightly more general result.
We first identify some properties satisfied by all simple \TCWs generated by a
\MPDA, then we show that all simple \TCWs satisfying these properties have
bounded split-width.
Considering $Y=X \cup \{\zeta\}$, 

\begin{restatable}{lemma}{lemMPDAbound}\label{MPDA-bound}
  A $k$-round well-timed simple \TCW using $n>1$ stacks has split-width at most
  $K=\max(kn+2(kn-1)|Y|,(k+2)(2|Y|+1))$.
\end{restatable}

Notice that if we have only one stack ($n=1$) hence also one round ($k=1$) the 
bound $4|Y|+2$ on split-width was established in Claim~\ref{lem:tpda}.

\begin{proof}
  Again, the idea is to play the split-game between \emph{Adam} and \emph{Eve}.
  \emph{Eve} should have a strategy to disconnect the word without introducing
  more than $K$ blocks.  The strategy of \emph{Eve} is as follows: Given
  the $k$-round word $w$, \emph{Eve} will break it into $n$ split-\TCWs.  The
  first split-\TCW only has stack 1 edges, and the second has stack edges
  corresponding to stack 2, etc.\ up to the last split-\TCW for stack $n$.  Each
  split-\TCW will have at most $k|Y|$ holes and can now be dealt with as we did
  in the case of TPDA.   The only thing to calculate is the number of cuts required in isolating each  split-\TCWs, which we do in the rest of the proof below.

  \subparagraph{\bf{Obtaining the $n$ split-\TCWs containing only stack $p$
  edges}} Since we are dealing with $k$-round \TCWs, we know that the stack
  operations follow a nice order : stacks 1, \dots, $n$ are operated in order
  $k$ times.  More precisely, by \eqref{item:Tp1}, the set $V$ of points of the
  simple \TCW \stcw can be partitioned into $V=V_1^1\uplus\cdots\uplus V_1^n
  \uplus\cdots\uplus V_k^1\uplus\cdots\uplus V_k^n$ such that
  $V_1^1<\cdots<V_1^n<\cdots<V_k^1<\cdots<V_k^n$ and $V_1^p\cup\cdots\cup V_k^p$
  contains all and only stack operations from stack $p$.  \emph{Eve}'s strategy
  is to separate all these sets, i.e., cut just after each $V_i^p$.  This
  results in $kn$ blocks.  This will not disconnect the word if there are edges
  with clock timing constraints across blocks, i.e., from some $V_i^p$ to some
  other block $V_j^q$ on the right with $p\neq q$.
  
  Fix some clock $x\in Y$ and assume that there are timing constraints which are
  checked in block $V_j^q$ and are reset in some block on the left.  All these
  crossing over timing constraints for clock $x$ come from some block $V_i^p$
  containing the last reset block of clock $x$ on the left of $V_j^q$.  With at
  most two cuts, we detach the consecutive sequence of resets $R_j^q(x)$ which
  are checked in $V_j^q$.  Doing this for every clock $x\in Y$ and every block
  $V_j^q$ except $V_1^1$, we use at most $2(kn-1)|Y|$ cuts.  So now, we have
  $kn+2(kn-1)|Y|$ blocks.  

  The split-\TCW $\stcw^q$ for stack $q$ consists of the reset blocks of the
  form $R_j^q(x)$ ($1\leq j\leq k$, $x\in Y$) together with the split blocks
  $W_j^q=V_j^q\setminus\bigcup R_i^p(y)$ where the union ranges over $i\geq j$,
  $p\neq q$ and $y\in Y$.  The split-\TCW $\stcw^q$ is disconnected from the
  other split-\TCWs $\stcw^p$ with $p\neq q$.  Moreover, we can check that, by
  removing the reset blocks $R_i^p(y)$ with $p\neq q$, we have introduced
  at most $k'|Y|$ holes in $W_1^q,\cdots,W_k^q$, where $k'=k$ if $q<n$ and 
  $k'=k-1$ if $q=n$. So the number of blocks of $\stcw^q$ is at most 
  $k|Y|+k+k'|Y|\leq k(2|Y|+1)$.
  
  Notice that if $k=1$ and $n>1$ then we have at most $2|Y|+1$ blocks in 
  $\stcw^q$ with $q<n$ and $|Y|+1$ blocks in $\stcw^n$.

  We apply now the TPDA game on each split-\TCW $\stcw^q$ using at most $4|Y|+2$
  extra cuts.  We obtain a split bound of $K=\max(kn+2(kn-1)|Y|,(k+2)(2|Y|+1))$.
\end{proof}

Having established a bound on the split-width for \MPDA restricted to $k$
rounds, we now discuss the construction of tree automata \BKMV and \AKMS 
 when the underlying system is a \MPDA. As a first step, we
keep track of the current round (context) number in the finite control.  This makes sure that the tree automaton only accepts runs using at most $k$-rounds. The validity check (and hence, \BKMV)  is as before.   
The only change pertains to the automaton that checks correctness of the underlying run, namely, $\AKMS$, as we need to handle $n$ stacks (and $k$ rounds) instead of the single stack.  

\subsection*{Tree Automata Construction for MultiStack and Complexity}
\label{app:mpdatree}

Given a \MPDA $\Sys=(S, \Sigma, \Gamma, X, s_0, F, \Delta)$, we first
construct a \MPDA $\Sys'$ that only accepts runs using at most $k$-rounds.
The idea behind constructing $\Sys'$ is to easily keep track of the $k$-rounds
by remembering in the finite control of $\Sys'$, the current round number and
context number.  The initial states of $\Sys'$ is $(s_0,1,1)$.  Here 1,1
signifies that we are in round 1, and context 1 in which operations on stack 1
are allowed without changing context.  The states of $\Sys'$ are
$\{(s,i,j) \mid 1 \leq i \leq k, 1 \leq j \leq n, s \in L\}$.

Assuming that $(s, \gamma, a, \op, R, s') \in \Delta$ is a transition of 
\Sys, then the transitions $\Delta'$ of $\Sys'$ are as follows:
\begin{enumerate}
  \item $((s,i,j), \gamma, a, \op, R, (s',i,j))\in\Delta'$ if 
  $\op$ is one of $\nop$, $\push^j_c$ or $\pop^j_{c\in I}$,

  \item $((s,i,j), \gamma, a, \op, R, (s',i,h))\in\Delta'$ if
  $j<h$ and $\op$ is one of $\push^h_c$ or $\pop^h_{c\in I}$,

  \item $((s,i,j), \gamma, a, \op, R, (s',i+1,h))\in\Delta'$ if
  $h<j$ and $\op$ is one of $\push^h_c$ or $\pop^h_{c\in I}$.

\end{enumerate}
The final states of $\Sys'$ are of the form $\{(s,i,j) \mid s\in F\}$.  It
can be shown easily that accepting runs of $\Sys'$ correspond to accepting 
$k$-round bounded runs of \Sys.

Now, given the \MPDA $\Sys'$, the tree automaton \BKMV that checks
for validity and realizability are exactly as in Theorem \ref{thm:AKMV}.

The only change pertains to the automaton that checks correctness of the
underlying run.  The tree automaton for the underlying system, \AKMS stores the set $E^s$ of pairs $(i,j)$
of left endpoints of \emph{distinct} blocks such that there is at least one
push-pop edge pertaining to stack $s$ from $\mathsf{Block}(i)$ to $\mathsf{Block})(j)$.   The other change is in 
the transitions of \AKMS : the transitions   $\add{i}{j}{\matchrel I}((i,a^1)\oplus(j,a^2))$
 in Table \ref{tab:AKMSb} will now differ, since this checks the correctness 
 of well-nesting of stack edges when we have only one stack.  This transition will be replaced with 
  $\add{i}{j}{\matchrel I,s}((i,a^1)\oplus(j,a^2))$, which also talks about the particular stack $s$. 
Thus, 
instead of one set $E$ as in the case of TPDA, we require $n$ distinct sets to
ensure well-nesting for each stack edge.  Secondly, we need to ensure that the
$k$-round property is satisfied.  Rather than doing this at the tree automaton
level, we have done it at the \MPDA level itself, by checking this in $\Sys'$.
This blows up the number of locations by $nk$, the number of stacks and
rounds.  Thus, the number of states of the tree automaton \AKMS that checks
correctness when the underlying system is a $k$-round \MPDA is hence 
$\leq (nk|\Sys|)^{\cO(K)}\cdot 2^{\cO(nK^2(|X|+1))}$, where $K=(4nk+4)(|X|+2)$.

\begin{restatable}{proposition}{propAKMSMP}\label{prop:AKMSMP}
  Let \Sys be a $k$-round multistack timed automaton of size $|\Sys|$ (constants
  encoded in unary) with $n$ stacks and set of clocks $X$.  Then, we can build a
  tree automaton \AKMS of size $(nk|\Sys|)^{\cO(K)}\cdot 2^{\cO(nK^2(|X|+1))}$
  such that $\Lang(\AKMS) {=} \{\tau{\in}\Lang(\AKMV) {\mid}
  \sem{\tau}\in\STCW(\Sys)\}$.
\end{restatable}

Using arguments similar to  Theorem \ref{thm:final-sys}, we obtain
\begin{theorem}
\label{thm:final}
  Checking emptiness for $k$-round \MPDA is decidable in  \textsc{ExpTime}.
\end{theorem}

%%% Local Variables:
%%% mode: latex
%%% TeX-master: "main.tex"
%%% End:

\section{Conclusion}
\label{sec:disc}

The main contribution of this paper is a technique for analyzing timed systems
via tree automata.  This is a new approach, which is quite different from all
existing approaches, most of which go via the region/zone based techniques for
timed automata.  The hardest part of our approach, i.e., checking realizability
is oblivious of the underlying timed system and hence we believe that this
technique can be applied uniformly to a variety of timed systems.  While, we
have made a few simplifying assumptions to best illustrate this method, our
technique can easily be adapted to remove many of these.  For instance, diagonal
constraints of the form $x-y\in I$ can be handled easily by adding matching
edges.  For a constraint $x-y\in [1,5]$, we add an edge between the last reset
of $x$ and last reset of $y$ with $[1,5]$ interval.  Thus, we can check the
diagonal constraint at the time of last reset.

As future work, we would also like to extend our results to other restrictions
for $\MPDA$ such as bounded scope and phase.  This would only require us to
prove a bound on the split-width and modify the system automaton appropriately
to handle the abstract behaviors generated by such systems.  Our techniques
could also be applied to the more general model \cite{hscc15} of recursive
hybrid automata.  Another interesting future work is to use our technique to go
beyond reachability and show results on model checking for timed systems.  While
model-checking against untimed specifications is easy to obtain with our
approach, the challenge is to extend it to timed specifications.  We also see a
strong potential to investigate the emptiness problem for classes of alternating
timed automata and hybrid automata.  At the very least, we would obtain new
width-based under-approximations which may lead to new decidability results.

%%% Local Variables:
%%% mode: latex
%%% TeX-master: "main.tex"
%%% End:

\bibliographystyle{plain}
\bibliography{ref}

\end{document}